\documentclass[preprint,authoryear,11pt]{elsarticle}

\usepackage{natbib}
\usepackage{epsfig}
\usepackage{graphicx}
\usepackage{bm}
\usepackage{amsmath,amssymb,amsthm,amsfonts,mathrsfs,textcomp,mathtools}
\usepackage{url}
\usepackage{multirow}
\usepackage{subfig}
\usepackage{aas_macros}

\usepackage{fullpage}

\theoremstyle{plain}
\newtheorem{thm}{Theorem}
\newtheorem{lemma}{Lemma}

\let\notORI\not 
\usepackage[mathb]{mathabx}
\let\not\notORI

\usepackage{tikz,pgfplots}
\usetikzlibrary{arrows,intersections,backgrounds}
\usetikzlibrary{positioning,calc,plotmarks}

\graphicspath{{figures/}}

\def\R{\mathbb{R}}

\newcommand{\clomon}{{\sc clomon}} 
\newcommand{\calR}{\mathcal{R}}
\newcommand{\calS}{\mathcal{S}}
\newcommand{\calI}{\mathcal{I}}

\journal{Icarus}

\begin{document}

\title{Completeness of Impact Monitoring}
\author[pi,sds]{Alessio Del Vigna}
\ead{delvigna@mail.dm.unipi.it}
\author[]{Andrea Milani $\dagger$\footnote{$(\dagger)$ During the final revision of this
    paper a tragedy occurred: Prof. Andrea Milani passed away on
    November 28, 2018. The present paper is dedicated to his memory.}}
\author[imcce]{Federica Spoto}
\author[sds]{Andrea Chessa}
\author[inaf,ifac]{Giovanni B. Valsecchi}

\address[pi]{Dipartimento di Matematica, Universit\`a di Pisa, Largo
  Bruno Pontecorvo 5, Pisa, Italy}
\address[sds]{Space Dynamics Services s.r.l., via Mario Giuntini,
  Navacchio di Cascina, Pisa, Italy}
\address[imcce]{IMCCE, Observatoire de Paris, PSL Research University,
  CNRS, Sorbonne Universités, UPMC Univ. Paris 06, Univ. Lille, 77
  av. Denfert-Rochereau F-75014 Paris, France}
\address[inaf]{IAPS-INAF, via Fosso del Cavaliere 100, 00133 Roma,
  Italy}
\address[ifac]{IFAC-CNR, via Madonna del Piano 10, 50019 Sesto
  Fiorentino, Italy}

\begin{small}
    \begin{abstract}
        The completeness limit is a key quantity to measure the
        reliability of an impact monitoring system. It provides the
        impact probability threshold for which every virtual impactor
        (VI) with impact probability above this value has to be
        detected. The completeness limit depends on the confidence
        region sampling: a goal of this paper is to increase the
        completeness without increasing the computational load, thus
        we propose a new method to sample the Line Of Variations (LOV)
        with respect to the previously one used in NEODyS. The
        step-size of the sampling is not uniform in the LOV parameter,
        since the probability of each LOV segment between consecutive
        points is kept constant. Moreover, the sampling interval has
        been extended to the larger interval $[-5,5]$ in the LOV
        parameter and a new decomposition scheme in sub-showers and
        sub-returns is provided to deal with the problem of duplicated
        LOV points appearing in the same return.

        The impact monitoring system \clomon-2 has been upgraded with
        all these new features, resulting in a decrease of the impact
        probability $IP^*$ corresponding to the generic completeness
        limit by a factor $\simeq 4$ and in an increase of the
        computational load by a factor $\simeq 2$. Moreover, since the
        generic completeness limit is an analytic approximation, we
        statistically investigate the completeness actually reached by
        the system. For this we used two different methods: a direct
        comparison with the results of the independent system Sentry
        at JPL and an empirical power-law to model the number of
        virtual impactors as a function of the impact probability. We
        found empirically that the number of detected virtual
        impactors with $IP>IP^*$ appears to grow according to a
        power-law, proportional to $IP^{-2/3}$. An analytical model
        explaining this power-law is currently an open problem, but we
        think it is related to the way the number of virtual impactors
        within a time $t_{rel}$ from the first observed close approach
        accumulates. We give an analytical model and we prove that
        this cumulative number grows with a power-law proportional to
        $t_{rel}^3$. The power-law allows us to detect a loss of
        efficiency in the virtual impactors search for impact
        probabilities near the generic completeness limit. The outcome
        of the comparison with Sentry shows that the two histograms of
        the number of VIs as a function of $IP$ are very consistent
        for $IP>IP^*$, which supports the confidence in the
        reliability of both systems.
    \end{abstract}
\end{small}
\maketitle

\begin{small}{\noindent\bf Keywords}: Impact Monitoring, Generic
    completeness, Line Of Variations\end{small}

\section{Introduction}
\label{sec:intro}

Some asteroids with an Earth-crossing orbit may impact our planet. A
crucial issue is to be able to identify the cases that could have a
threatening Earth close encounter within a century, as soon as new
asteroids are discovered or as new observations are added to prior
discoveries. The main goal of impact monitoring is to solicit
astrometric follow-up to either confirm or more likely dismiss the
announced risk cases, \emph{i.e.}, asteroids having some virtual
impactors \citep{milani:visearch}. This is achieved by communicating
the impact date, the impact probability and the estimated impact
energy.

This activity requires an automated system that continually monitors
the Near-Earth Asteroids (NEAs)
catalog. \clomon-2\footnote{\url{http://newton.dm.unipi.it/neodys/index.php?pc=4.1}}
and Sentry\footnote{\url{http://cneos.jpl.nasa.gov/sentry/}} are two
independent impact monitoring systems that have been operational at
the University of Pisa since 1999 and at JPL since 2002, providing the
list of asteroids with a non-zero probability to impact the Earth
within a century \citep{milani:clomon2}. There is a constant
comparison between the results of the two systems and, as required by
the International Astronomical Union, the results are carefully
cross-checked before any public announcement of an impact risk above
an agreed level, as measured by the Palermo Scale
\citep{chesley:2002PS}. There is a probability threshold called
\emph{generic completeness limit} that the two systems set as a
goal. Above this threshold the search for impact possibility has to be
complete, that is every virtual impactor (\emph{i.e.}, each connected
set of initial conditions leading to a collision with a planet) with
an impact probability greater than the completeness limit has to be
detected. Desirably, the generic completeness levels of the two
concurrent impact monitoring systems need to be as close as possible,
in order to have a common threshold down to which the two systems can
be compared.

Since the generic completeness limit is a theoretical quantity,
defined under some simplified assumptions, the level of completeness
actually reached by the system has to be measured \emph{a
  posteriori}. If the generic completeness limit in impact probability
is somewhat lower than the actual level achieved, it means that there
is a loss of efficiency in the VI search, that is some VI which could
in theory be detected is missed in the scan. Finding possible causes
and trying to decrease the number of missed VIs leads to an
improvement of the whole system, filling as much as possible the gap
between the two completeness limits. There are two methods to measure
this quantity: the first is based on an empirical law to model the
number of virtual impactors as a function of the impact probability;
the second is a direct comparison with the results of an other
independent system, namely Sentry, since we do not have a ``ground
truth'', that is a practical way to generate an absolutely complete
list of all possible VIs above a given impact probability $IP$. We
analyzed the results of the application of our new method by
exploiting both the techniques.

\section{The impact monitoring problem}
\label{sec:IM}

The mathematical methods used in impact monitoring have been developed
over the years, in a sequence of papers to which we refer the reader
for a complete explanation: \cite{milani:AN10},
\cite{milani:visearch}, \cite{chesley:2002PS},
\cite{valsecchi:resret}, and \cite{milani:clomon2}.

The classical procedure to determine the orbit of an asteroid uses as
initial condition at $t_0$ the solution $\mathbf{x}^*\in \R^N$ of a
non-linear least squares fit, along with its covariance matrix
$\Gamma\coloneq \Gamma(\mathbf{x}^*)$
\cite[Chapter~5]{milani:orbdet}. We denote with $N$ the dimension of
the parameter space\footnote{The fit parameters space has dimension
  $N=6$ when we solve for the six orbital elements, but it could have
  higher dimension if some other parameter is determined along with
  the orbital elements. A common situation is the determination of the
  Yarkovsky-related semimajor axis drift \citep{farnocchia:yarko,
    chesley:yarko, delvigna:yarko}. This has consequences also for
  impact monitoring. Indeed, a dynamical model including
  non-gravitational forces is sometimes needed to make reliable impact
  predictions, especially if the hazard analysis time span is extended
  to time intervals longer than one century.}. The nominal solution is
surrounded by a set of orbits that are still compatible with the
observational dataset, the so-called \emph{confidence region}. It is
the basic tool for the impact monitoring activity, since impact
predictions have to take into account the nominal solution as well as
its uncertainty, given by the matrix $\Gamma$. The confidence region
can be approximated by the \emph{confidence ellipsoid}
\[
    Z_{lin}(\sigma) \coloneqq \left\{\mathbf{x}\in \R^N \,:\,
    (\mathbf{x}-\mathbf{x}^*)^T C (\mathbf{x}-\mathbf{x}^*) \leq
    \sigma^2\right\},
\]
where $\sigma>0$ is a confidence level and $C\coloneq \Gamma^{-1}$ is
the normal matrix. Thus the confidence ellipsoid $Z_{lin}(\sigma)$ is
the region delimited by the $(N-1)$-dimensional ellipsoid given by the
positive definite quadratic form associated to the normal matrix
$C$. As explained in what follows, the confidence ellipsoid is just
used for local computations, since the assumptions to use this
approximation are not applicable to impact monitoring in general.

\subsection{Sampling of the confidence region}

The main goal of impact monitoring is to establish whether the
confidence region contains virtual impactors. Thus the confidence
region is sampled by a finite set of \emph{Virtual Asteroids}
(VAs). Since the dynamical system describing the asteroid orbits is
not integrable, only a finite number of VA initial conditions can be
computed and propagated over the selected time interval. This sampling
has to be done in an efficient way, that is with a few but selected
orbits, in such a way that they are as much as possible representative
of the infinite set of possible orbits. The geometric sampling methods
are one possible way to select the ensemble of virtual asteroids: for
this class of methods the sampling takes place on the intersection
between the confidence region and a differentiable manifold. In
particular, \cite{milani:ident1} and \cite{milani:multsol} introduced
a 1-dimensional sampling method, in which the geometric object is a
smooth line in the orbital elements space, the Line Of Variations. The
main advantage of this approach is that the set of VAs has a geometric
structure, that is they belong to a differentiable curve along which
interpolation is possible.

Another sampling method, namely Monte Carlo, directly uses the
probabilistic interpretation of the least squares principle, sampling
the probability distribution in the orbital elements space to obtain a
set of equally probable orbits \citep{chodas96}. More complex sampling
methods, such as 2-dimensional ones, have been proposed in
\cite{tommei:2d} and have been recently used to deal with the problem
of the imminent impactors \citep{farnocchia2015, spoto:immimp}.

\subsection{LOV propagation}

The LOV sampling computation provides a set of orbits
$\{\mathbf{x}(\sigma_i)\}_{i=-M,\,\ldots,\,M}$, corresponding to
values $\{\sigma_i\}_{i=-M,\,\ldots,\,M}$ of the LOV parameter. As
introduced in \cite{milani:AN10}, the second step in impact monitoring
consists in the propagation of each VA in the future. The close
approaches of each VA with the Earth are recorded by means of sets of
points on the Target Plane (TP), which is the plane passing through
the Earth's center and orthogonal to the unperturbed velocity of the
asteroid\footnote{That is, orthogonal to the incoming asymptote of the
  hyperbola defining the two-body approximation of the trajectory at
  the time of closest approach \citep{valsecchi:resret}.}. To avoid
geometric complications, we consider ``close'' only those approaches
with a distance from the Earth center of mass not exceeding some value
$R_{TP}$. Possible values for $R_{TP}$ range between $0.05$~au and
$0.2$~au, thus the target planes are in fact disks with a finite
radius.

The first output of the LOV propagation is a collection of Earth
encounters that have been detected for each VA during the time span of
interest. Each close approach with initial condition $\mathbf{x}_i$ is
represented by at least one trace $\mathbf{y}_i=(\xi_i,\zeta_i)$ on
the corresponding target plane\footnote{For a discussion on the choice
  of the coordinates $(\xi,\zeta)$ on the TP, see
  \cite{valsecchi:resret,milani:clomon2}.}.

\subsection{Stretching and width}
\label{subsec:stretching}
Let us denote with $\mathbf{g}:\R^N\rightarrow \R^2$ the function
mapping an orbit $\mathbf{x}$ at epoch $t_0$ to the point $\mathbf{y}$
on the TP of an encounter occurring at an epoch $t_1$. This map is the
composition of the propagation from $t_0$ to $t_1$ and of the
projection on the target plane. In the linear approximation, which is
justified since the return analysis is performed locally, the
confidence ellipsoid $Z_{lin}^X(\sigma)$ around $\mathbf{x}^*$ is
mapped onto a confidence ellipse $Z_{lin}^Y(\sigma)$ around
$\mathbf{y}^*=\mathbf{g}(\mathbf{x}^*)$. By applying the covariance
propagation law, the target plane ellipse is defined by
\[
  (\mathbf{y}-\mathbf{y}^*)^T C_{\mathbf{y}} (\mathbf{y}-\mathbf{y}^*)
  \leq \sigma^2,
\]
where $C_{\mathbf{y}}=\Gamma_{\mathbf{y}}^{-1}$ and
$\Gamma_{\mathbf{y}} = D_{\mathbf{x}^*}\mathbf{g}
\,\Gamma_{\mathbf{x}}\, (D_{\mathbf{x}^*}\mathbf{g})^T$.  The square
roots of the eigenvalues of $\Gamma_{\mathbf{y}}$ are the semimajor
and semiminor axis of $Z_{lin}^Y(1)$, respectively called
\emph{stretching} and \emph{width}.

If we have a LOV sampling, as in the case of impact monitoring, we
have a further map $\mathbf{s}:\R\rightarrow \R^N$, which is the
parameterization of the LOV as a differentiable curve, that is
$\mathbf{s}(\sigma_i)=\mathbf{x}_i$. In this way we can consider the
composite map $\mathbf{f}\coloneq \mathbf{g}\circ \mathbf{s}$ from the
sampling space to the target plane and define the \emph{stretching
  along the LOV} in $\sigma$ as
\begin{equation}\label{eq:stretching}
  S(\sigma) \coloneq \left| \frac{d\mathbf{f}}{d\sigma}(\sigma) \right|.
\end{equation}
The stretching along the LOV measures the displacement of two points
on the target plane as a function of the separation between the
corresponding points in the sampling space.

\subsection{Return analysis}
\label{subsec:retanalysis}

If the close encounter collection is not empty, the first step in the
decomposition is to sort the encounters by date, followed by a
splitting into \emph{showers} that are clustered in time. The showers
are then divided into contiguous LOV segments: the shower encounters
are re-sorted according to their LOV index, and then, advancing
through the sorted list, we ``cut'' the shower wherever a gap in the
LOV index is encountered. We call \emph{returns} these sets of
dynamically related encounters. After the entire propagation, each
close approach is carefully analyzed to search for virtual impactors.

When there are many points on the TP in a given return it is easy to
understand the LOV behavior: the stretching is small and the linear
theory is locally applicable. On the contrary, in strong non-linear
cases the stretching is large and changes rapidly from point to point:
in this case a local analysis is necessary in the neighborhood of each
VA. We refer to \cite{milani:clomon2,tommei:phd} for a discussion on
the possible geometries of the LOV trace on a target plane and for a
proper solution for each case. Here we only want to outline the basic
idea of the return analysis with a simple example. The key point is
that the virtual asteroids are not just a set of points but they
sample a smooth curve, allowing us to interpolate between consecutive
sample points. For instance, let us suppose two consecutive VAs
$\mathbf{x}_i$ and $\mathbf{x}_{i+1}$ have TP trace points
$\mathbf{y}_i$ and $\mathbf{y}_{i+1}$ straddling the Earth impact
cross section. If the geometry of the TP trace is simple enough
(principle of simplest geometry), an interpolation method provides a
point on the LOV $\mathbf{x}_{i+\delta}$ with $0<\delta<1$ and such
that $\mathbf{y}_{i+\delta}$ is inside the Earth impact cross section:
then, around $\mathbf{x}_{i+\delta}$ there is a virtual impactor. If a
virtual impactor has been found, by computing the probability density
function with a suitable Gaussian approximation centered at
$\mathbf{x}_{i+\delta}$ it is possible to estimate the probability
integral on the impact cross section, that is the impact probability
associated with the given VI.

\subsection{Generic completeness definition}
\label{sec:generic_comp}

This paper is focused on the completeness of impact monitoring, that
is on the completeness of the VI search. The \emph{completeness limit}
can be formally defined as the highest impact probability VI that can
escape the detection. Given the complexity of the problem of impact
monitoring, this completeness cannot be computed and thus we use an
approximate definition, which assumes idealized circumstances. The
\emph{generic completeness limit} is the highest impact probability VI
that could possibly escape detection, if the associated return on the
target plane is fully linear \citep{milani:clomon2}, that is under the
hypothesis of full linearity of the map $\mathbf{f}$. Under this
generic assumption the trace of the VAs on the target plane is simply
a straight line on the TP: if there is a VI, this line intersects the
impact cross section $D$ on the TP in a chord of the circle bounding
$D$.

For a VI to be detectable by the system, at least one LOV orbit has to
cross the target plane. For now, we make the assumption that one point
on the TP is sufficient for the VI detection (see
Section~\ref{sec:miss_VI} for a discussion on this choice). The
stretching is a key quantity to estimate the number of LOV orbits that
intersects the target plane of a given encounter: the higher is the
stretching, the greater is the separation between two consecutive
points on the TP. In particular, if the stretching becomes too high,
the separation on the TP could exceed the diameter $2R_{TP}$ of the TP
itself, thus no point crosses the TP and the virtual impactor is not
detected with certainty. As a consequence, to have at least one point
on the TP, the stretching cannot exceed a maximum value
$S_{max}$. Assuming a uniform step-size $\Delta \sigma$ for the LOV
sampling, the condition $S \cdot \Delta\sigma \leq 2R_{TP}$ must hold,
which implies
\begin{equation}\label{stretch_cond}
    S\leq \frac{2 R_{TP}}{\Delta\sigma} \eqcolon S_{max}.
\end{equation}
We now convert the previous inequality into a condition concerning the
impact probability. The probability density function induced on the
LOV by the probability density function on the orbital elements
space\footnote{The probability density defined on the orbital elements
  space is the propagation of the Gaussian density assumed for the
  residuals, and it is Gaussian with mean $\mathbf{x}^*$ (the nominal
  solution) and covariance matrix $\Gamma = \Gamma(\mathbf{x}^*)$, the
  covariance matrix of the orbit determination least squares fit
  \citep{milani:orbdet}.} is
\begin{equation}\label{eq:gauss_LOV}
    p(\sigma) \coloneq \frac{1}{\sqrt{2\pi}} e^{-\frac{\sigma^2}2},
\end{equation}
where $\sigma$ is the LOV parameter. Under our assumptions, the impact
probability of the VI is given by integrating $p(\sigma)$ over the
inverse image of the diametrical chord contained in the LOV projection
on the TP. Assuming that the stretching is $S$, the integration domain
is an interval with length ${2b_\Earth}/{S}$ in the $\sigma$-space. As
a consequence, the following estimation holds:
\begin{equation*}
    IP \simeq \frac{2b_\Earth}{S} \cdot p(0) \geq
    \frac{1}{\sqrt{2\pi}} \frac{2b_\Earth}{S_{max}} =
    \frac{\Delta\sigma}{\sqrt{2\pi}} \cdot \frac{b_\Earth}{R_{TP}}.
\end{equation*}
where $b_\Earth$ is the radius of the Earth impact cross section on
the TP, which takes into account the gravitational focusing
\citep{valsecchi:resret}. Since the generic completeness limit $IP^*$
is the minimum impact probability of a VI for which condition
\eqref{stretch_cond} is satisfied, we have
\[
    IP^* \coloneq \frac{\Delta\sigma}{\sqrt{2\pi}} \cdot \frac{b_\Earth}{R_{TP}}.
\]
This quantity depends on the amount of gravitational focusing, and can
be higher for asteroids with a low velocity at infinity. To obtain a
uniform threshold, we can use a typical value for $b_\Earth$,
\emph{e.g.}, $b_\Earth= 2R_\Earth$. This is an approximation, but we
need to chose a fixed value applicable to all asteroids with VIs, thus
we are using a value appropriate for low relative velocity NEA, taking
into account that these have larger probability of having VIs. In this
way we have an approximated value for the generic completeness limit:
\begin{equation}\label{eq:comp_limit}
    IP^* = \frac{\Delta\sigma}{\sqrt{2\pi}} \cdot
    \frac{2R_\Earth}{R_{TP}}.
\end{equation}

With the previous simple equation we can provide an estimate of the
completeness reached by \clomon-2 while it was operating with a
uniform LOV sampling. As shown in \cite{milani:clomon2}, we
obtain\footnote{The TP radius adopted for the close approaches
  detection was $R_{TP}=0.2\text{ au }\simeq 4700\,R_\Earth$, the LOV
  sampling was performed with 2401 virtual asteroids over the interval
  $|\sigma|\leq 3$, thus $\Delta\sigma =0.0025$.}
\[
    IP^* \simeq 4.24\cdot 10^{-7},
\]
corresponding to a maximum stretching value of $S_{max} = 3.8\cdot
10^6\,R_\Earth$. In the case at least two points on the TP are
required for the detection of a VI, the generic completeness limit
$IP^*$ would be simply twice as much.

\section{An optimal method for LOV sampling}
\label{sec:prob_sample}

Before the switch to the new method presented in this paper, \clomon-2
performed the LOV sampling by means of uniformly spaced points in the
parameter $\sigma$, and over the interval $|\sigma|\leq 3$. Since the
probability density on the LOV is the Gaussian defined by
\eqref{eq:gauss_LOV}, a uniform step in $\sigma$ is not optimal
because the probability of each sampling interval is high around
$\sigma = 0$, whereas it becomes too low near the LOV endpoints. This
is the reason to use a step-size that is inversely proportional to the
probability density. The new sampling is such that the probability of
the interval among two consecutive points of the sampling is
constant. It means that if $\{\sigma_i\}_{i=1,\,\ldots,\,M}$ are the
sampling nodes, then
\[
    \mathbb P([\sigma_i,\sigma_{i+1}]) \coloneq
    \int_{\sigma_i}^{\sigma_{i+1}} p(\sigma)\,d\sigma
\]
is constant, \emph{i.e.}, it does not depend on $i$. As a consequence,
the sample points will be more dense around the value $\sigma=0$
(nominal solution) whereas they become more sparse moving towards the
tips. Furthermore, to avoid too long intervals at LOV tails, when the
interval length exceeds a certain threshold $\Delta\sigma_{max}$ the
sampling becomes uniform to the threshold value. This technical detail
is needed to bound the step-size, since too large values of
$\Delta\sigma$ could result in divergence of the algorithms to find
VIs, effectively cutting the LOV tails out from the analysis. To avoid
this loss and also to keep a better control of the cases in which a
geometrically large VI occurs for large values of $\sigma$, the
sampling interval has been extended to $|\sigma|\leq \sigma_{max}=5$
(and also to cover the same interval as JPL's Sentry system).

We now derive a condition for the step-size to guarantee a constant
probability to each sampling interval. This is achieved by repeating
the same argument of Section~\ref{sec:generic_comp}, not around the
nominal solution, but around the VA corresponding to a generic
$\sigma$ value. In particular, to have at least one point on the TP
corresponding to the LOV parameter $\sigma$, the stretching must
satisfy the condition
\[
    S(\sigma) \leq \frac{2R_{TP}}{\Delta\sigma(\sigma)}\eqcolon
    S_{max}(\sigma).
\]
Under linearity assumptions, the impact probability of the VI around
the value $\sigma$ of the LOV parameter is given by
\[
    IP(\sigma) \simeq \frac{2b_\Earth}{S(\sigma)} \cdot p(\sigma) \geq
    \frac{2b_\Earth}{S_{max}(\sigma)} \cdot p(\sigma).
\]
By assuming $b_\Earth=2R_\Earth$ as above, we can define
\begin{equation}\label{eq:compl_dep_on_sigma}
    IP^*(\sigma)\coloneq \frac{2
    R_\Earth}{R_{TP}}\Delta\sigma(\sigma)\*p(\sigma),
\end{equation}
which is the minimum probability of a detectable VI around the value
$\sigma$ of the LOV parameter. To ensure the detection of a VI with
probability $IP^*(\sigma)$ for all $|\sigma|\leq \sigma_{max}$, we
have to define the generic completeness as
\[
    IP^* \coloneq \sup_{|\sigma|\leq \sigma_{max}} IP^*(\sigma).
\]
Notice that we can take the supremum since $IP^*(\sigma)$ are bounded
from above. Using \eqref{eq:compl_dep_on_sigma}, this implies the
following condition on the step-size:
\[
    \Delta\sigma(\sigma) \leq \frac{R_{TP}}{2R_\Earth} IP^*
    \frac{1}{p(\sigma)}.
\]
The computation of $\Delta\sigma_i\coloneq \Delta\sigma(\sigma_i)$
starts from the value $\sigma_0=0$ (corresponding to the nominal
solution), and by recursion we compute
\begin{equation}\label{eq:step-size}
    \left\{\begin{array}{l} \Delta\sigma_{i}=
          \min\left(\dfrac{R_{TP}}{2R_\Earth} IP^* \dfrac
            1{p(\sigma_{i})},
            \Delta\sigma_{max}\right),\quad i\geq 0\\[0.4cm]
          \sigma_{i+1}=\sigma_i+\Delta\sigma_i,\quad i\geq 1
      \end{array}\right.
\end{equation}
for the sampling of the interval $0\leq \sigma \leq \sigma_{max}$. For
the negative side, that is the interval $-\sigma_{max}\leq \sigma\leq
0$, the nodes are $\{-\sigma_i\}_{i\geq 0}$. By definition of $IP^*$,
the step-size is such that the inverse image of the diametrical chord
has the same probability for all $\sigma$. And this in turn implies
that each sampling interval has the same probability. Of course this
holds only for the intervals with length not exceeding
$\Delta\sigma_{max}$.

\begin{figure}[h!]
    \centering
    \includegraphics[width=0.75\textwidth,scale=0.3]{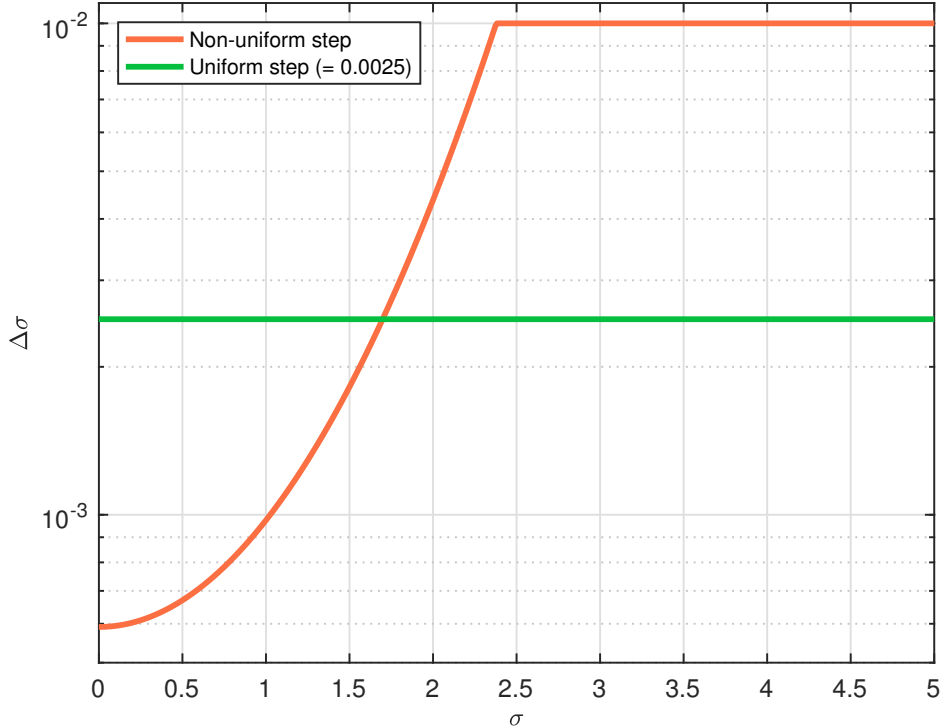}%
    \caption{Graph of the step-size as a function of the LOV parameter
      $\sigma$, with the parameter choice as in \eqref{eq:opt}.
      \emph{Orange line}: step-size for the uniform-in-probability
      sampling. \emph{Green line}: uniform step-size, as previously
      used by \clomon-2.}\label{fig:step-size}
\end{figure}

The new method allows one to choose the completeness limit before the
sampling procedure starts. Thus the number of VAs per LOV side is
known only at the end of the procedure \eqref{eq:step-size}, since it
stops when $\sigma_{max}$ is exceeded. This is somewhat different with
respect to the previous method, for which we first established the
number of VAs per LOV side and as a consequence the completeness level
was determined, as in
equation~\eqref{eq:comp_limit}. Figure~\ref{fig:step-size} shows the
behavior of the step-size as a function of the LOV parameter
$\sigma$. It was generated using the following values, which are the
same used for the current impact monitoring computations:
\begin{equation}\label{eq:opt}
  IP^* = 1\cdot10^{-7},\quad \sigma_{max}=5,\quad \Delta\sigma_{max}=0.01.
\end{equation}

This choice of parameters leads to the computation of at most 4719
multiple solutions\footnote{The number of VAs may actually be lower,
  because we terminate the sampling when the residuals become too
  large, currently when $\chi>5$.}, whereas at most 2401 VAs were
computed with the previous sampling method
\citep{milani:multsol}. This gives about twice the computational load
than before, that however permits a decrease by a factor $\simeq 4$ in
the generic completeness limit, since the value corresponding to the
previous uniform sampling was $IP^*\simeq 4.24\cdot 10^{-7}$, as
computed in Section~\ref{sec:generic_comp}.

\section{Missing VI detection: possible causes}
\label{sec:miss_VI}

Once a generic completeness level $IP^*$ has been established, the
goal is to find all the virtual impactors with probability
$IP>IP^*$. This might not be achieved in actual computations in case
some VI is not detected by the scan. In general, the identification of
the causes and the development of possible solutions are important
issues. We were aware of two possible sources of VI loss: in what
follows we discuss both of them.

\subsection{Duplicated points in the same return}

In Section~\ref{subsec:retanalysis} we have defined showers and
returns as particular dynamically related subsets of the set of close
encounters of all the virtual asteroids. The associated iterative
procedure properly works as long as the showers are
well-defined. Indeed, there are cases in which there is not a clear
clustering in time among the encounters, causing the presence of very
long showers (also called extended showers) in the decomposition. As a
consequence, some virtual asteroids appear multiple times in the same
return, and this must be avoided for the subsequent TP analysis to be
successful. In general, several phenomena may cause this problem, for
instance temporary capture of the asteroid by the Earth, Earth-like
orbit (as for 2000~SG$_{344}$) and encounters with low relative
velocity. There are also cases in which a close approach, defined in
time as the interval in which the distance from the encountered body
is less than some $D_{min}$ (for the Earth $D_{min}=0.2$~au), contains
multiple occurrences of a local minimum distance, and this also
generates returns with duplicated points. Such bad cases are not so
rare as one may think, especially if we use a denser LOV sampling. The
problem affects $\simeq 25\%$ of the asteroids in the NEODyS risk list
(as of April 2018).

An example of such situation is given by asteroid 2000~SG$_{344}$,
already used to present the problem of duplicated virtual asteroids in
\cite{milani:clomon2}, Figure~6. We show a figure which cannot be
identical to the one in \cite{milani:clomon2} since, even if the
observational data set is exactly the same, it is currently treated
with a different astrometric error
model. Figure~\ref{fig:shower_2000SG344_bef} shows a single extended
shower for 2000~SG$_{344}$ lasting for about one year. In this
situation the previous algorithm identifies a single shower around the
year 2069. Figure~\ref{fig:shower_2000SG344_bef} shows the closest
encounter date of the points belonging to the extended shower, against
the LOV index. The clustering into returns is clear from the picture,
but there are returns with multiple occurrence of the same virtual
asteroid, as highlighted in orange.

\begin{figure}[h!]
    \centering
    \includegraphics[width=0.7\textwidth,scale=0.3]{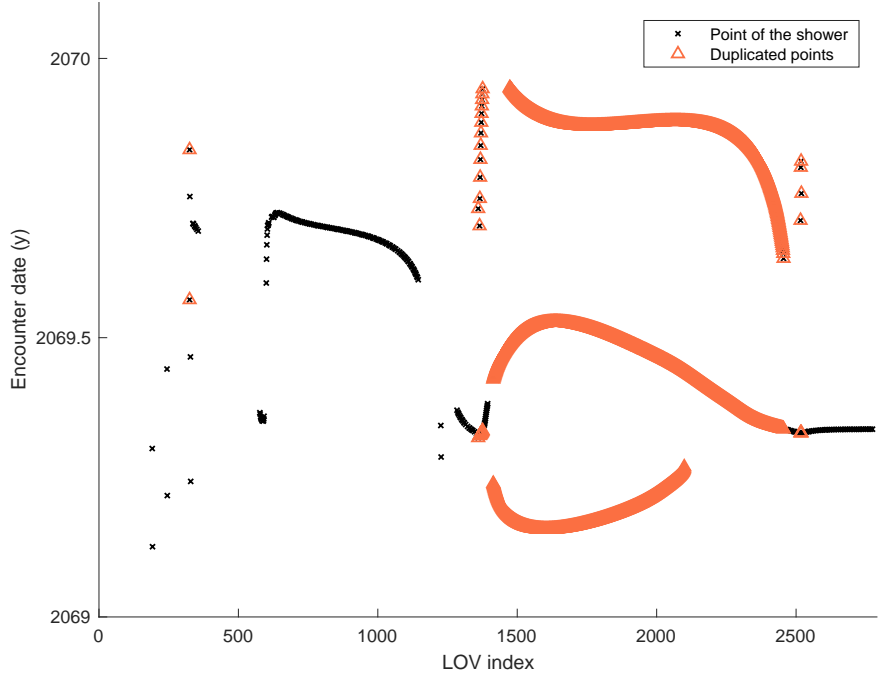}%
    \caption{Extended shower around the year 2069 for asteroid
      2000~SG$_{344}$. For each virtual asteroid belonging to the
      shower, we plot the closest encounter date against the LOV index
      (black crosses). The virtual asteroids that appear more than one
      times in the same return are highlighted (orange
      triangles).}\label{fig:shower_2000SG344_bef}
\end{figure}

To handle such cases, we decided to introduce a further
splitting procedure every time a return contains duplicated virtual
asteroids. This leads to the definition of sub-showers and
sub-returns. First, we sort the return by ascending closest approach
time, then we scan and divide it as follows: we cut the return every
time a virtual asteroid is already present among the previous ones,
starting from the previous cut. Each subset obtained in this way is
called a \emph{sub-shower}. Second, each sub-shower is divided
into contiguous LOV segments, called \emph{sub-returns}, in the same
way we obtain the returns from the showers.  \ref{app:dec_subret}
contains a mathematical description of the procedure just described,
with a formal proof of completion. Thus the decomposition algorithm
ensures that each sub-return is free from duplicated virtual
asteroids. Then we use, as returns of the original shower, all the
sub-returns of all the sub-showers. Figure~\ref{fig:split_ret} shows
the outcome of the application of the splitting procedure to the
return going from LOV index 1414 to 2779. As a graphical
representation of the algorithm, we pass a horizontal line from the
bottom to the top of the plot to scan the return by ascending closest
approach time and we make a cut when we encounter a duplication. The
return is decomposed into three sub-showers (represented with
different marks and colors), each of which is further decomposed in
contiguous LOV segments.

It is apparent from Figure~\ref{fig:split_ret} that our algorithm
splits more than the minimum possible, \emph{e.g.}, in this figure it
can be seen that there are two cuts splitting dynamically related
encounters (the second and third horizontal lines starting from the
bottom). However, this has no negative consequences on the performance
of the VI detection, because even the LOV interval between the last
index of a sub-return and the first of the next is actually scanned,
by using the algorithms for the tail and head of the return
\citep{milani:clomon2}. The split is different from the one performed
by Sentry, but the result in terms of impact monitoring is the same.

In the work carried out for the switch to the new sampling of the Line
Of Variations, we implemented the splitting procedure into sub-showers
and sub-returns, thus the results presented in
Section~\ref{sec:results} take already into account also this
improvement.

\begin{figure}[h!]
    \centering
    \includegraphics[width=0.70\textwidth]{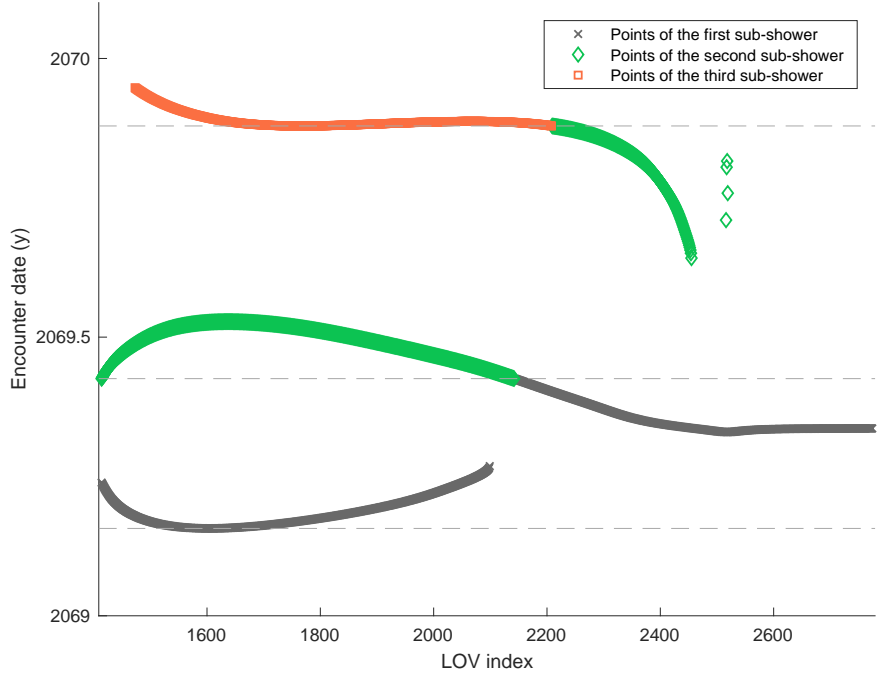}%
    \caption{Application of the decomposition procedure to a return of
      the extended shower of 2000~SG$_{344}$. The return is decomposed
      into three sub-showers, represented with different marks and
      colors (orange, green, and gray, respectively). We mark with a
      black circle the first point (in terms of time) of each
      sub-shower: in particular, the first LOV indices of the three
      sub-showers are 1604, 1414 and 1766,
      respectively.}\label{fig:split_ret}
\end{figure}

\subsection{Extreme non-linear cases}
When there is a strong non-linearity due to previous close approaches,
the stretching is large and rapidly varying, causing the LOV behavior
to be complex. Thus strong non-linearity of the map $\mathbf{g}$
introduced in Section~\ref{subsec:stretching} can lead to unsuccessful
detections of virtual impactors.

Inside a return, only some intervals between consecutive VAs can
contain a minimum of the closest approach distance and they are
identified by a geometric classification \citep{milani:clomon2}. The
analysis continues by checking if the minimum distance could be small
and, if so, by applying iterative schemes and interpolation along the
LOV to determine the minimum distance and the corresponding LOV orbit,
respectively. In most cases, both \clomon-2 and Sentry use the
\emph{modified regula falsi} applied to the continuous function
\[
    f(\sigma) \coloneq \frac{d r^2}{d\sigma}(\sigma)
\]
over the interval $[\sigma_1,\sigma_2]$ under consideration, where
$r^2(\sigma) \coloneq \xi^2(\sigma) + \zeta^2(\sigma)$ is the square
of the distance from the Earth center. This algorithm is convergent,
but failures may occur if for some value of $\sigma$ in the interval
the function is undefined. It happens when the TP of the encounter
around that date is missed, which generally indicates that the two TP
points under consideration do not actually belong to the same return
\citep{milani:clomon2}.

Another typical situation that may determine an unsuccessful detection
is caused by singletons. They are returns consisting of one single
point on the target plane, due to a very high value of the stretching
at the corresponding LOV point and indicating an extremely non-linear
situation. By definition the modified \emph{regula falsi} cannot be
applied in this case. The solution adopted by \clomon-2 is the Newton
method with bounded steps \citep{milani:clomon2}: this method cannot
diverge, but can fail to converge either by finding a value of
$\sigma$ for which the TP is missed, as in the previous case, or by
exceeding a preset maximum number of iterations without achieving
convergence with the required accuracy. In both cases there is the
possibility that a VI actually exists in the analyzed LOV segment, but
the method fails in detecting it. A possible solution to deal with
this problem is to resort to a densification technique. If we suitably
densify the LOV sampling around the orbit corresponding to the
singleton and obtain a return with $4$-$5$ points on the target plane
instead of a lone point, this makes the TP analysis easier and more
effective. Indeed, Figure~\ref{fig:hist2} shows that the power-law
fits well only for values of $IP$ corresponding to at least $4$ points
on the TP: this is a further motivation for the use of a densification
technique also for returns with very few points, although singletons
remain the primary motivation for densification. We did not implement
the densification of the sampling, though we intend to include it in
our future work as discussed in Section~\ref{sec:conc}.

\section{Results}
\label{sec:results}

After the switch to the new sampling method, the actual level of
completeness reached by the system has to be measured and compared
with the one corresponding to the previous method. To perform this
analysis we make use of the histograms of the number of virtual
impactors $\mathcal{N}$ as a function of the inverse of the impact
probability $IP$ (for the sake of clarity we used $\log_{10}(1/IP)$ on
the horizontal axis). We made a histogram for the ensemble of VIs
obtained using a uniform step-size in $\sigma$ (previous sampling
method employed by \clomon-2) and a second one for the results of the
new sampling method, uniform in probability, on the same set of
asteroids, with the observations available at the same date. We used
the data contained in the NEODyS
database\footnote{\url{http://newton.dm.unipi.it/neodys/}} immediately
before and after \clomon-2 switched to the new method, namely on 29
October 2016. As a sample, we used the 571 asteroids in the NEODyS
Risk List at that time, thus computed with the uniform sampling in the
LOV parameter $\sigma$. As a result of the application of the new
method we obtain a set of 558 asteroids with virtual impactors out of
571. The two histograms are shown in Figure~\ref{fig:hist1}
and~\ref{fig:hist2} (left panel). The first thing that stands out is
the very different total number of virtual impactors: the application
of the new method almost doubled this number, causing an increase from
13604 to 25942 virtual impactors.

In both the histograms there are two vertical lines corresponding to
two different values of impact probability:
\begin{itemize}
  \item the orange line (the one on the right) represents the impact
    probability corresponding to the assumption that a VI is
    detectable with at least one target plane point (which is the
    minimum requirement), in full linear conditions;
  \item the red line (the one on the left) represents the impact
    probability corresponding to the assumption that at least two
    target plane points are needed for the detection of a virtual
    impactor.
\end{itemize}
The part of the histogram on the left of the orange line corresponds
to $IP>IP^*$, namely to the impact probabilities of the virtual
impactors that the system should detect with certainty. The number of
virtual impactors $\mathcal{N}$ is expected to grow as the impact
probability goes down to $IP^*$. Even looking just to the histogram
bars we notice that the growth of $\mathcal{N}$ seems to slightly slow
down close to the vertical lines with respect to what one would
expect. To highlight this behavior we fitted the histogram contour
for $IP>IP^*$ with a suitable law: the best-fit line is represented in
light blue in both plots of Figure~\ref{fig:hist1} and~\ref{fig:hist2}
(left panel) and it was obtained with a linear correlation coefficient
$>0.99$ in both cases. In particular, we performed a linear fit of the
histogram contour in a log-log scale: Figure~\ref{fig:hist1}
and~\ref{fig:hist2} (right panel) show the points corresponding to the
histogram bar tips, those selected for the fit, and the best-fit
line. A linear fit in the log-log plot corresponds to a power-law for
the number of virtual impactors, that is
\begin{equation}\label{eq:ascending}
    \mathcal{N} = n(IP) = c_1 \cdot
    \left(\frac{IP^*}{IP}\right)^\alpha\quad \text{if }IP\geq IP^*.
\end{equation}
As a result we obtained the estimation $\alpha\simeq 0.678$ for the
histogram related to the uniform sampling, and $\alpha\simeq 0.664$
for the histogram related to the uniform-in-probability
sampling. These results are remarkably close, even if obtained with
different sampling of the LOV. Furthermore, other results of this
paper (see Section~\ref{sec:nVI_time} and Section~\ref{sec:JPL}) seem
to confirm this numerical evidence. At this point we are still not
able to provide a full interpretation of the value of $\alpha\simeq
2/3$ as a mathematical property of the ensemble of all the
VIs. Nevertheless, as a first step in this direction, in
Section~\ref{sec:nVI_time} we outline a possible model to explain the
growth of the number of virtual impactors $\mathcal{N}$ as a function
of the time.

\begin{figure}[t!]
    \centering
    \raisebox{-0.5\height}{
      \includegraphics[width=0.575\columnwidth]{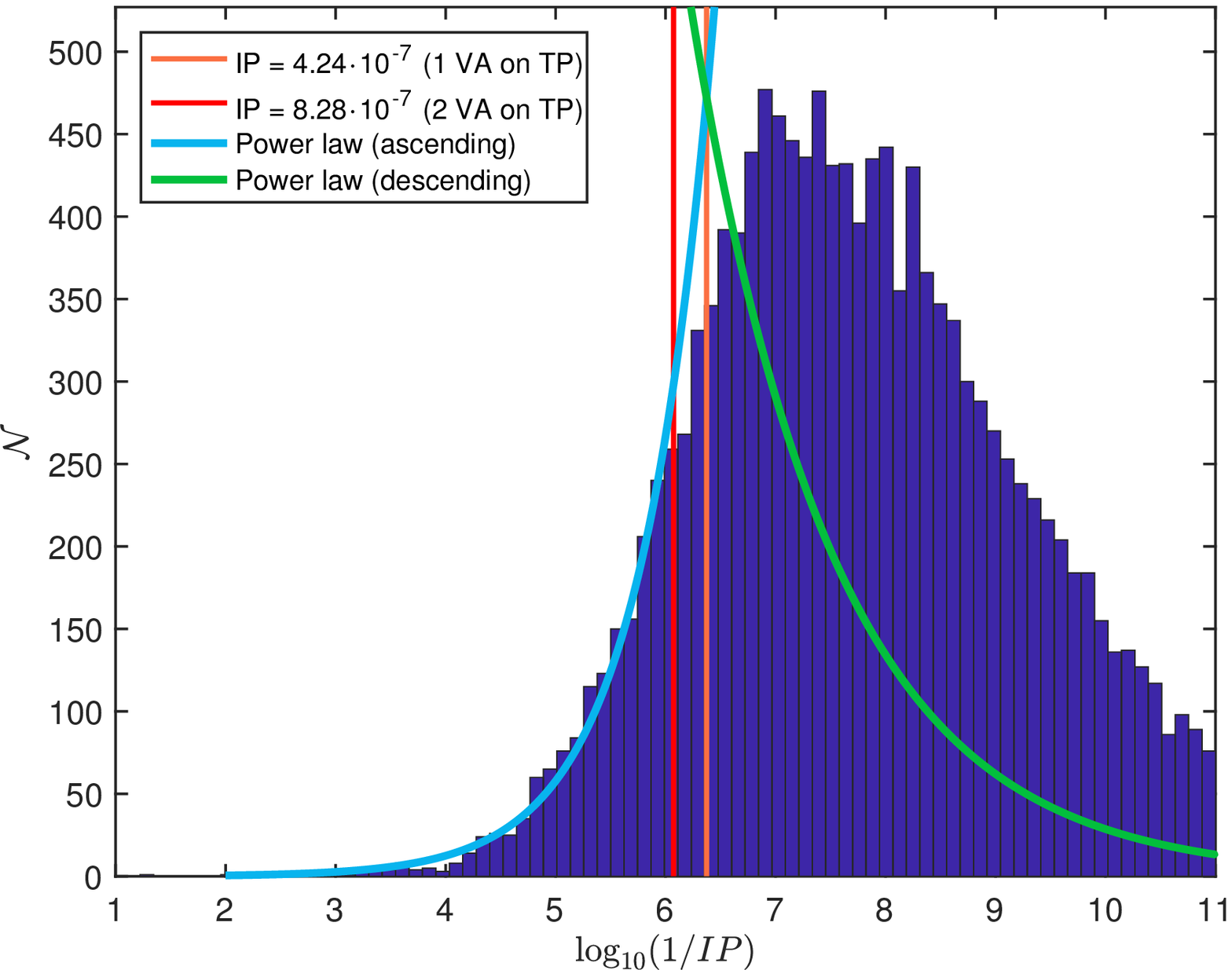}
    }%
    \hfill
    \raisebox{-0.5\height}{
      \includegraphics[width=0.375\columnwidth]{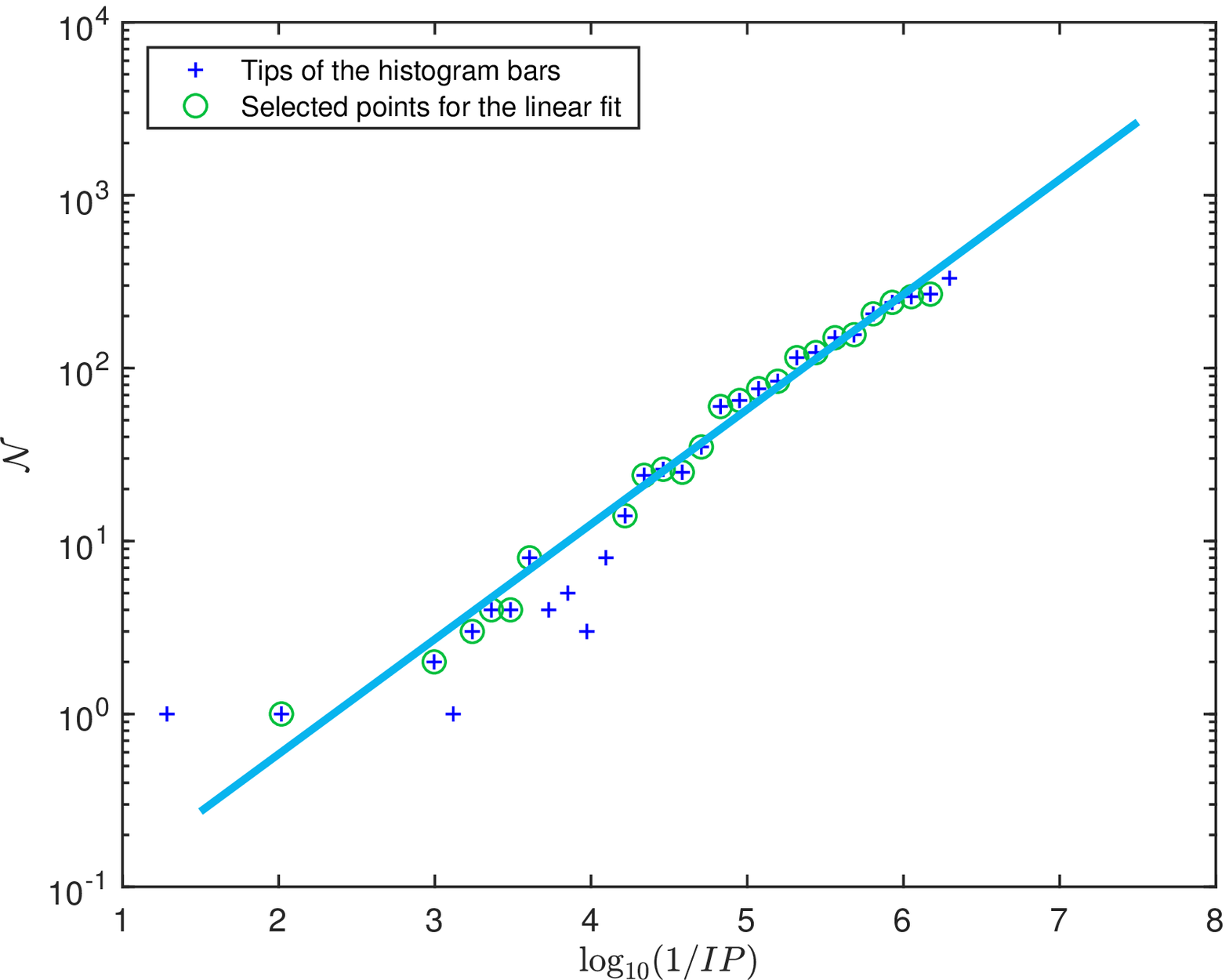}
    }%
    \caption{\emph{Left panel}. Histogram of the number of virtual
      impactors (as of October 2016) as a function of the inverse of
      the impact probability, in the case of uniform sampling in
      $\sigma$. \emph{Right panel}. Log-log plot of the histogram bar
      tips and corresponding regression line, providing
        $\alpha\simeq 0.678$.}
    \label{fig:hist1}
    \vspace{0.5cm}
    \raisebox{-0.5\height}{
      \includegraphics[width=0.575\columnwidth]{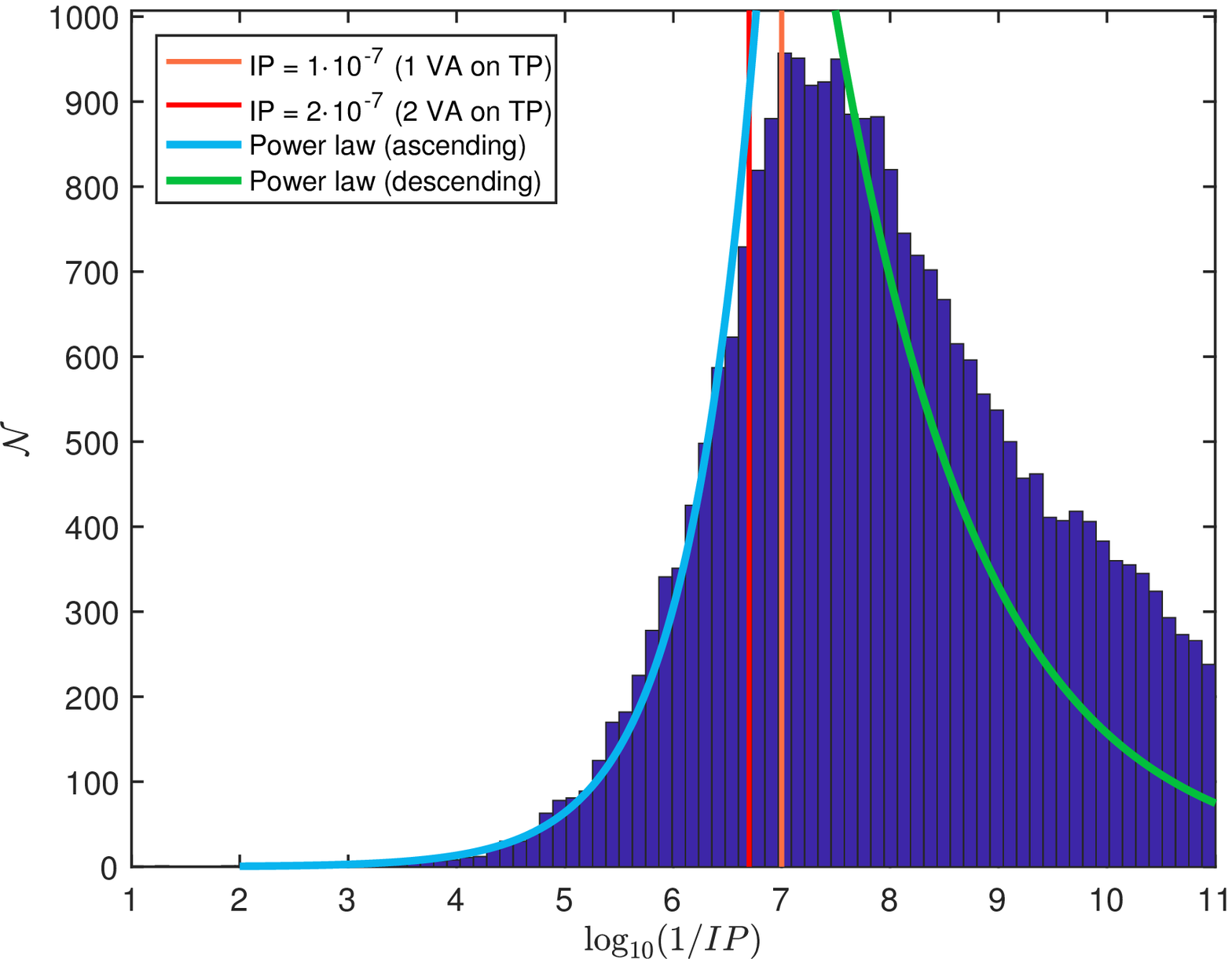}%
    }%
    \hfill
    \raisebox{-0.5\height}{
      \includegraphics[width=0.375\columnwidth]{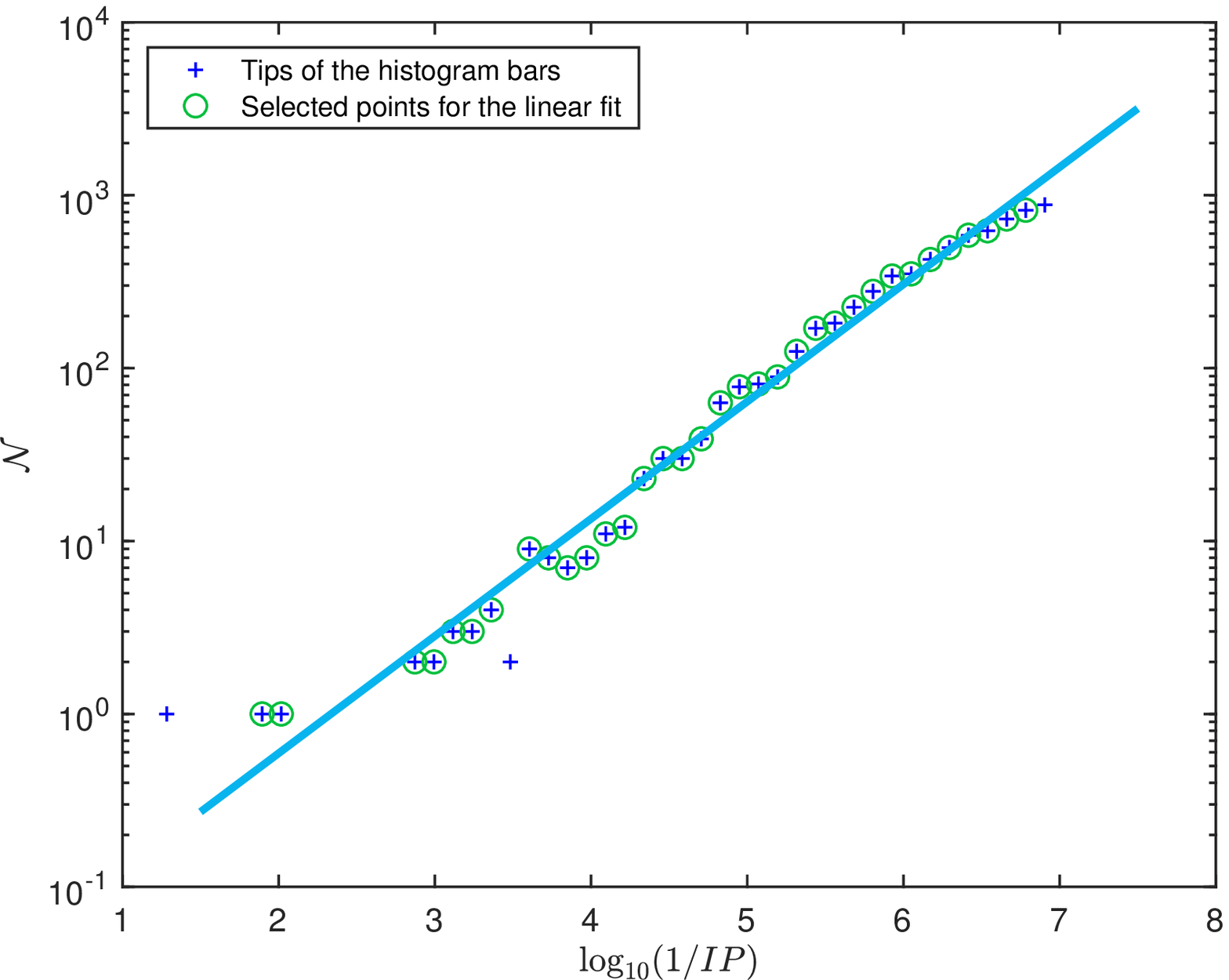}
    }%
    \caption{\emph{Left panel}. Histogram of the number of virtual
      impactors (as of October 2016) as a function of the inverse of
      the impact probability, in the case of uniform-in-probability
      sampling. \emph{Right panel}. Log-log plot of the histogram bar
      tips and corresponding regression line, providing
        $\alpha\simeq 0.664$.}
    \label{fig:hist2}
\end{figure}

For $IP<IP^*$, the probability to find a virtual impactor with impact
probability $IP$ is roughly the ratio $IP/IP^*$. Thus the number of
virtual impactor is
\[
    \mathcal{N} = n(IP) \* \frac{IP}{IP^*}= c_1 \cdot
    \left(\frac{IP^*}{IP}\right)^{\alpha-1}\quad \text{if }IP<IP^*.
\]
This equation corresponds to the descending green line in both
plots. We notice that this line does not fit the right side histogram
contour: there are more virtual impactors than expected from this
law. For the results obtained with the new sampling, this can be
explained by considering that the LOV sampling returns to be uniform
in $\sigma$ close to the LOV tails, and thus in that region the LOV
is over-sampled with respect to what would be needed to reach the
completeness level $IP^*$. With the old sampling, the production of
VIs with low impact probability is still larger, because the sampling
was not optimized specifically for $IP>IP^*$.

The differences between the fitted ascending curve corresponding to
equation~\eqref{eq:ascending} and the histogram clearly show that
there is a loss of efficiency in finding virtual impactors with impact
probability slightly above the completeness level. Indeed, for these
impact probability values the expected number of VIs based on the
empirically fitted power-law is larger than the number of actually
detected ones. To define the generic completeness limit we make the
assumption that even a single point on the target plane allows the
system to detect the virtual impactor, but from a practical point of
view this completeness level cannot be reached due to non-convergence
of the iterative schemes in some difficult cases, as explained in
Section~\ref{sec:miss_VI}. Actually, this does not happen only with a
single point on the target plane (singleton), but even with very few
points. As discussed in Section~\ref{sec:conc}, a densification of the
LOV sampling is the way to fill the gap between the actual
completeness level and the theoretical generic completeness. A
possible densification technique could convert returns with very few
points into return with 4-5 points. As a result, the iterative
methods used (such as \emph{regula falsi} and Newton's method
with bounded steps) should converge in a larger number of cases and
the VI search could be more efficient and complete.

\section{Analytical formulation for the time evolution}
\label{sec:nVI_time}

We analyze the behavior of the cumulative number of virtual impactors
$\mathcal{N}$ as a function of the time elapsed from the initial
conditions. As starting sample we used all the asteroids in the NEODyS
Risk List (as of April 2018), which contained 734 objects and 32906
virtual impactors, without considering the special cases\footnote{The
  four special cases are (101955)~Bennu, (99942)~Apophis,
  (29075)~1950~DA, and (410777)~2009~FD. These are currently the only
  asteroids that required the inclusion of the Yarkovsky effect for
  the impact monitoring.}. The histograms of Figure~\ref{fig:hist_i_H}
show the distribution of the inclination $i$ and of the absolute
magnitude $H$ among the considered objects. As clear from simple
arguments, the majority of the sample contains small low-inclination
asteroids: for instance, $95.5\%$ of the sample has absolute magnitude
$H>22$ and $71\%$ has inclination $i<5$~deg.

\begin{figure}[h!]
  \centering
  \includegraphics[width=0.475\columnwidth]{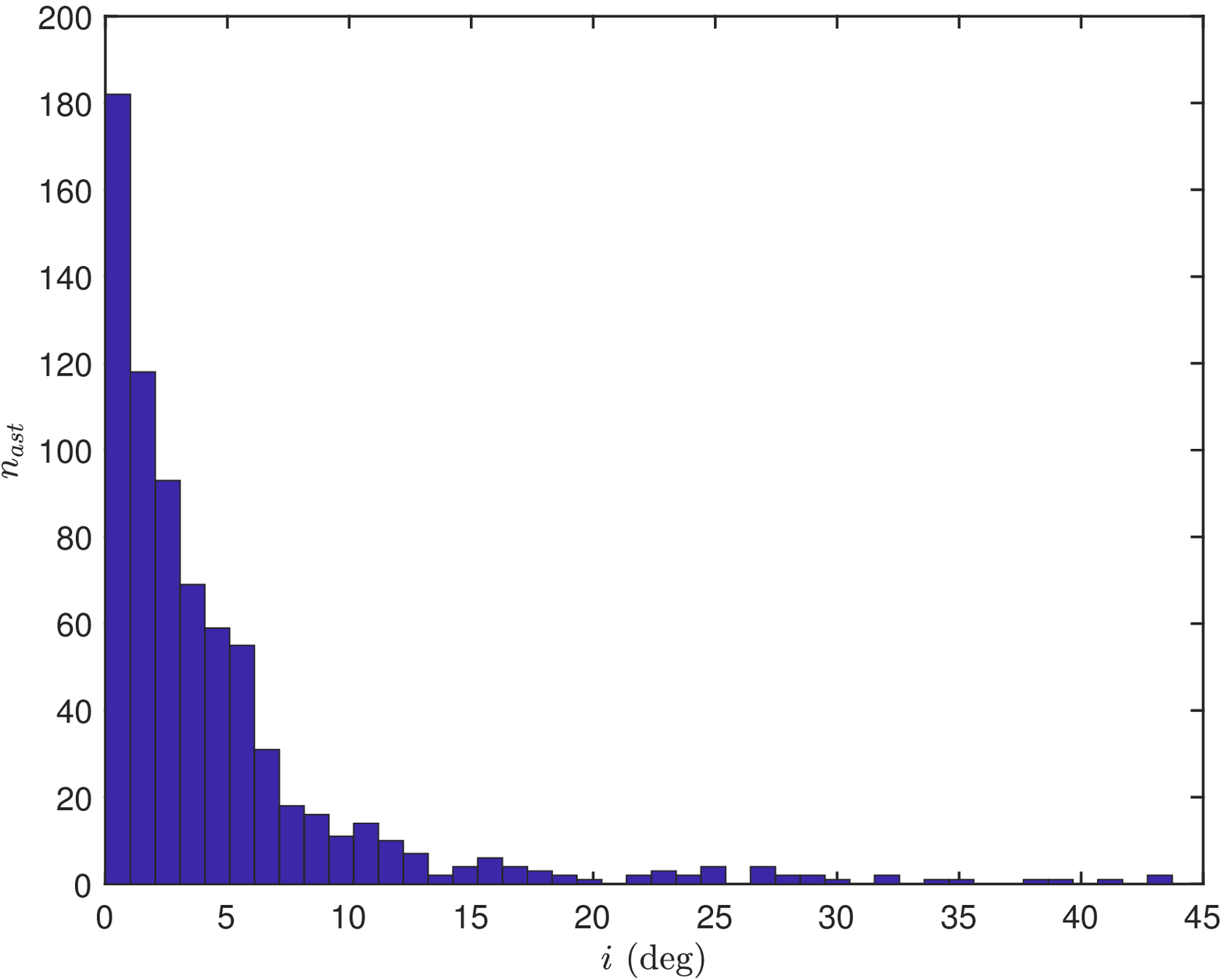}%
   \hspace{0.3cm}
  \includegraphics[width=0.475\columnwidth]{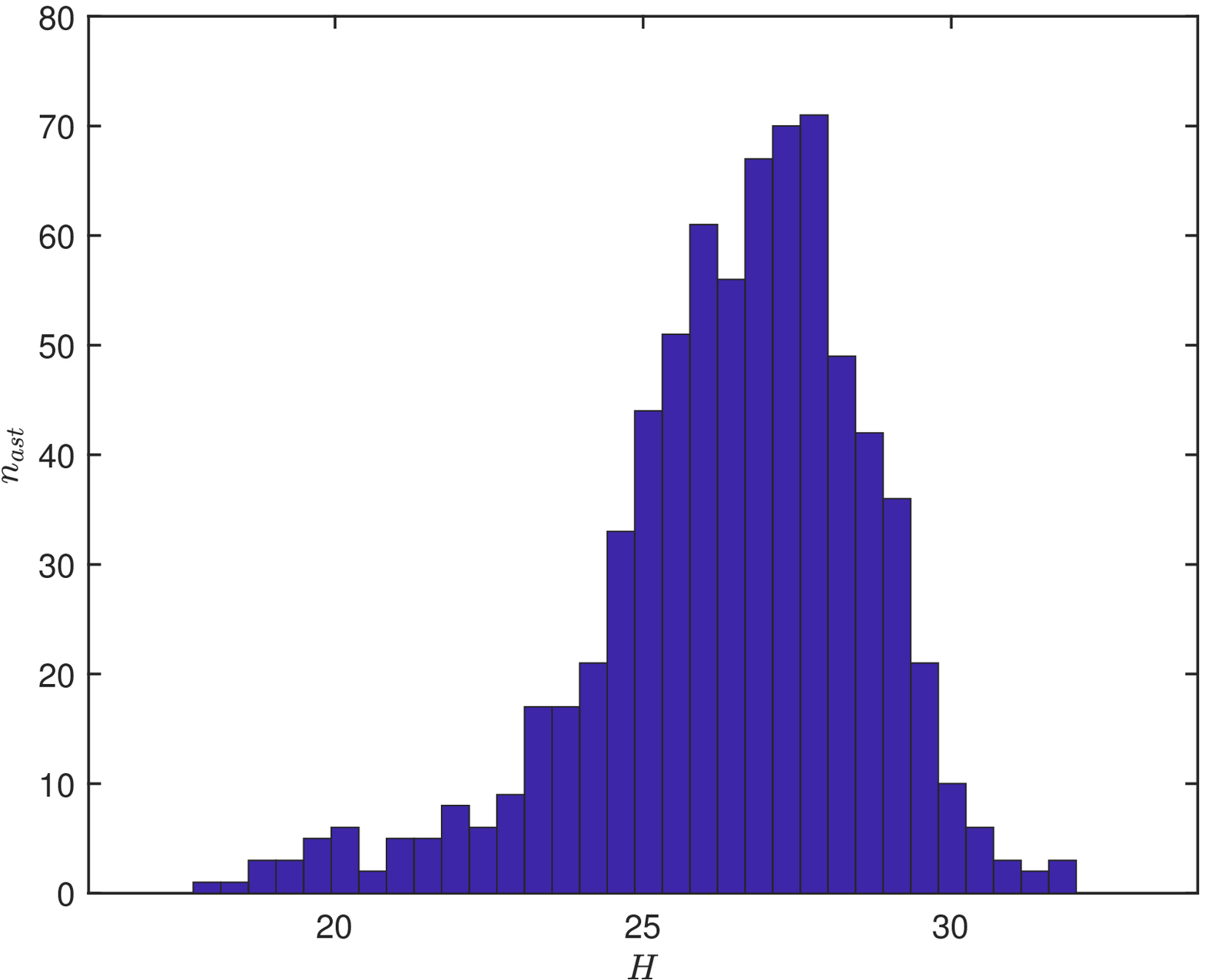}
  \caption{\emph{Left panel}. Histogram of the inclination $i$
    (deg). \emph{Right panel}. Histogram of the absolute magnitude
    $H$. Both the plots are referred to the set of $734$ asteroids in
    the NEODyS Risk List (as of April 2018).}
  \label{fig:hist_i_H}
\end{figure}

The sample of virtual impactors for the time evolution has to be
complete, that is it has to contain all the possible virtual impactors
with $IP$ down to a certain threshold. Thus the results of the new
sampling of the LOV are a good starting point, since the new method
ensures a complete virtual impactors search down to
$IP^*=1\cdot 10^{-7}$. To take into account the loss of completeness
due to singletons, as discussed in Section~\ref{sec:miss_VI}
and~\ref{sec:results}, we selected the virtual impactors with
$IP>2\cdot 10^{-7}$: the filtered set contains $6084$ virtual
impactors, corresponding to $473$ asteroids. We then applied a second
filter, considering the virtual impactors corresponding to
low-inclination asteroids, \emph{i.e.}, $i<5$~deg, since the discussion
below holds in an exact way in the planar case. In the end we analyzed
a sample of $5313$ virtual impactors with $IP>2\cdot 10^{-7}$ and
$i<5$~deg.

For a single asteroid, the accumulation of virtual impactors with time
depends on the time elapsed since the first observed close
approach. We call this relative time $t_{rel}$, assuming as origin
($t_{rel}=0$) the time of the first observed close approach. The exact
computation of $t_{rel}$ for each asteroid in the risk list would be
complicated, but, by taking into account Figure~\ref{fig:hist_i_H}
(right panel), we see that the vast majority of the asteroids in the
risk list is composed by very small objects. As a consequence, they
can only be discovered during a close approach. For almost all of the
asteroids in the risk list the center of the observed arc is thus a
good approximation of the origin $t_{rel}=0$. Thus, for each asteroid
in the risk list, we sorted the set of its virtual impactors by time
and we computed the relative time of each one of them with respect to
the center of the observed arc. Figure~\ref{fig:cum1} is the log-log
plot of the cumulative number of virtual impactors up to each value of
the relative time (blue and green marks). The log-log plot clearly
shows a linear growth, which we try to determine with a linear
fit. However, for a more accurate fit we have to cut out the tails of
the ostensible line. The tail for low relative times because its
contribution is weakened by small number statistics. In the tail for
high relative times the growth seems to slow down, but this is due to
the fact that the maximum relative time for which the scan for VIs has
been performed changes from asteroid to asteroid. Thus we performed
the linear fit over a suitable interval $t_1\leq t_{rel} \leq
t_2$. This fit corresponds to a power-law, that is
\begin{equation}
  \mathcal{N} = c_2 \cdot t_{rel}^\beta \quad \text{if }t_1\leq t_{rel} \leq t_2.
\end{equation}
Choosing $t_1=40$~yr and $t_2=99$~yr we obtained the estimation $\beta =
(3.001\pm 0.001)$, with a linear correlation coefficient $0.9994$. The
points marked with green circles in Figure~\ref{fig:cum1} are those
selected for the linear fit, and the orange straight line is the
resulting best-fit line.
The histogram in Figure~\ref{fig:cum2} is the histogram of the number
of virtual impactors accumulated up to the relative time. The plot
represents the same quantity of Figure~\ref{fig:cum1} with an
histogram, but not in a log-log scale. The colors have the same
meaning as for Figure~\ref{fig:cum1}: the green part corresponds to
the tail of the log-log plot used for the linear fit and the orange
line is the best fit power-law. Given the approximations introduced in
the model, this fit is remarkably good and identifies the power-law
proportional to $t_{rel}^3$ with very low uncertainty.

\begin{figure}[p!]
    \centering
    \includegraphics[width=0.85\columnwidth]{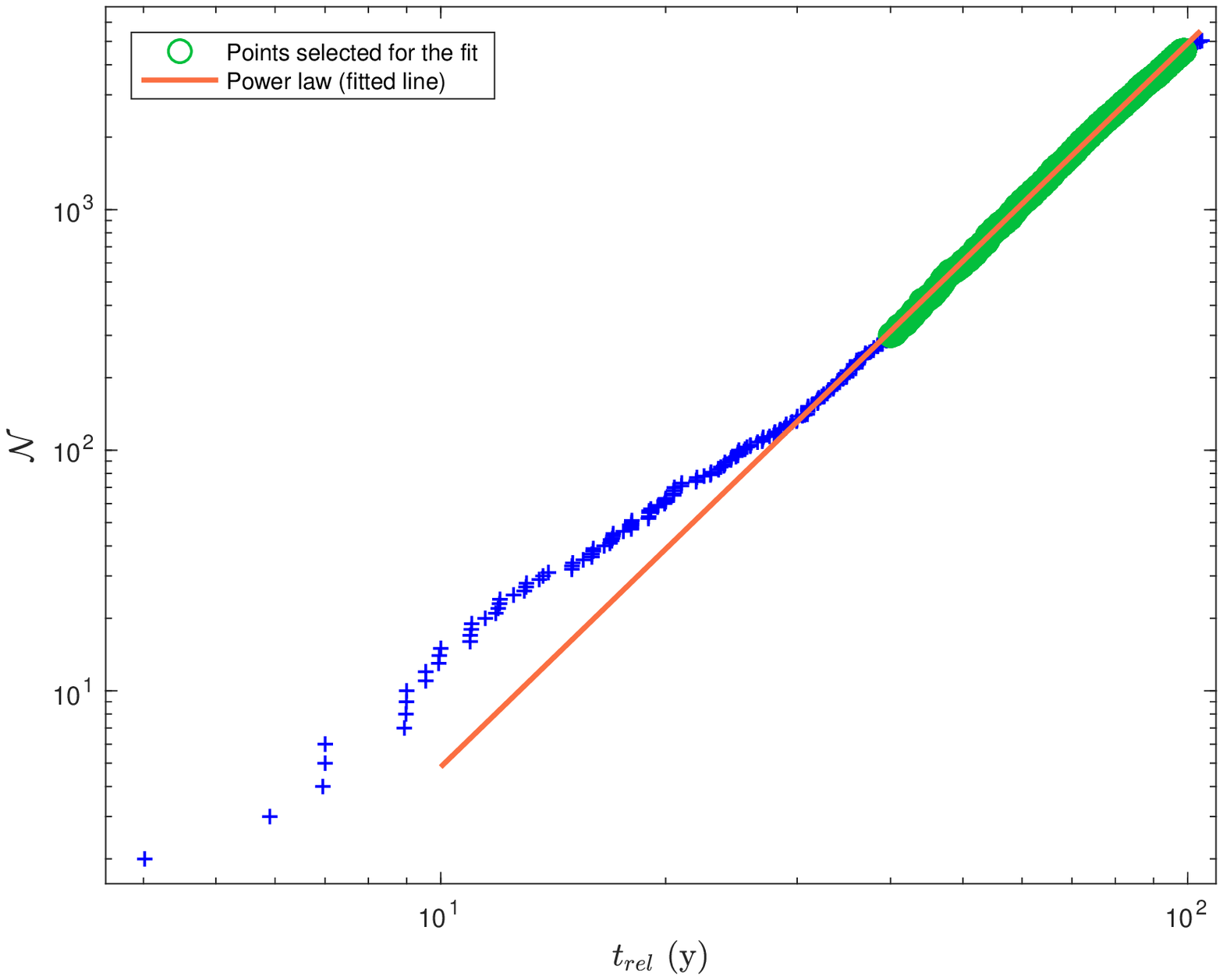}%
    \caption{Log-log plot of the cumulative number of virtual
      impactors $\mathcal{N}$ as a function of the relative time
      $t_{rel}$ (y). The points selected for the linear fit are marked
      with green circles. The orange straight line is the best-fit
      line obtained from the linear fit.}
    \label{fig:cum1}
    \vspace{0.7cm}
    \includegraphics[width=0.7\columnwidth]{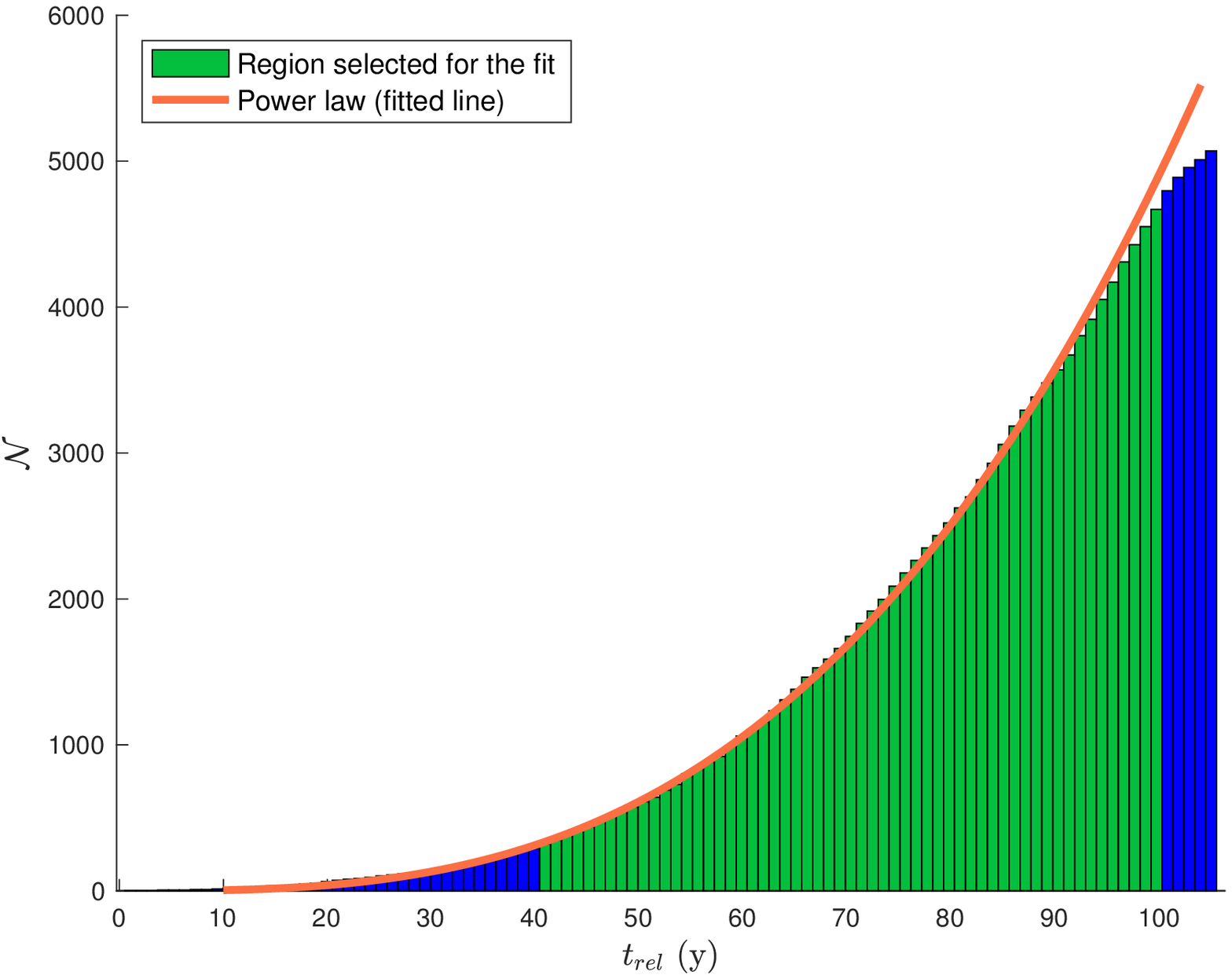}%
    \caption{Plot of the cumulative number of virtual impactors
      $\mathcal{N}$ as a function of the relative time $t_{rel}$ (y).}
    \label{fig:cum2}
\end{figure}

The power-law behavior with exponent $\simeq 3$ can be explained
using the analytical theory of close encounters as developed in
\cite{valsecchi:resret}, whose results agree with those of the
circular restricted three-body problem
\citep{valsecchi:cartography}. We consider the target plane
coordinates $\xi$ and $\zeta$: the former corresponds to the signed
local MOID, whereas the latter is related to the timing of the
encounter. We use the wire approximation
\citep{valsecchi:resret,milani:clomon2}, that is we assume that the
LOV projection on the target plane of a given encounter is a
continuous sequence of points, all with the same value of $\xi$ and
differing only for the value of $\zeta$.  Then, for each asteroid we
proceed along the same lines described in the appendix of
\cite{spoto:410777}.

\noindent The condition for a collision at a resonant return to take
place is that the ratio of the period of the small body and that of
the Earth is $k/h$, with $k$ and $h$ relatively prime integers. Then,
following the first encounter, after $h$ heliocentric revolutions of
the small body and $k$ revolutions of the Earth, both the Earth and
the small body will be back to the same position \citep{milani:AN10,
  valsecchi:resret}. This situation means that the post-encounter
semimajor axis $a'$ has to have precisely a certain value, say
$a'_\star$, with the corresponding mean motion $n'_\star$.  For the
Kepler's third law, the latter has to be
\begin{equation}\label{eq:resonance}
    n'_\star = \left(\frac{h}{k}\right)^{2/3}.
\end{equation}
As shown in \cite{valsecchi:resret} and in the appendix of
\cite{spoto:410777}, the values of $a'$ are constrained between a
maximum and a minimum, say $a'_{max}$ and $a'_{min}$, to which
correspond $n'_{min}$ and $n'_{max}$ respectively.  Consider now the
time interval in which we are interested: since $t_2 \simeq 99$~yr, it
is clear that we have to consider all the values of $n'_{min} \leq n'
\leq n'_{max}$ that, expressed as in \eqref{eq:resonance}, have
$k<99$. Thus the number of collision possibilities, \emph{i.e.}, of
virtual impactors, is proportional to the number of encounter
opportunities. This number accumulates in the same way as the number
of elements of $\mathcal{F}^{r,s}_n$, which is the set of irreducible
fractions between two integers $r$ and $s>r$, and whose denominators
do not exceed $n$, called Farey fractions
(see~\ref{app:farey}). Theorem~\ref{thm:farey4} states that the number
of elements of $\mathcal{F}^{r,s}_n$ grows like $n^2$ and that it
accumulates as $n^3$, which is the result highlighted from the fit of
Figure~\ref{fig:cum2}.

The above reasoning holds in an exact way, even for a single small
body, in the planar circular restricted three-body problem with Jacobi
constant $\mathcal{J}$ sufficiently high to ensure the small body will
not be expelled on a hyperbolic orbit, \emph{i.e.},
$\mathcal{J} > 2\sqrt{2}$ \citep{carusi82}. Of course the hypotheses
on which our analytical estimate is based are an approximation of the
more complex problem of asteroid close approaches, nevertheless the
general trend turns out to be confirmed by our statistical
analysis. If we add back to the list of asteroids with VIs (with
$IP>2\cdot 10^{-7}$) the ones with higher inclination, $i>5$~deg, the
fit for the slope in the log-log plot gives $\beta = 2.829$,
indicating that the model we have proposed can represent accurately
the statistics of the VIs time distribution only in the low
inclination case, as expected.

Note that the slope predicted by the Farey number-theoretical
arguments refers to the number of close approaches to the Earth, not
to the number of virtual impactors. The fact that the histogram
of VIs as a function of $t_{rel}$ follows the same power-law expected
for the number of close approaches indicates that, on average, the
number of virtual impactors is proportional to the number of
close approaches. This is by no means an obvious result. The
possibility of an impact during such a close approach is controlled by
the MOID (Minimum Orbital Intersection Distance) at the time of the
encounter: if the MOID is larger than the radius $b_\Earth$ of the
Earth impact cross section, collisions cannot occur. The MOID changes
both as a consequence of short-periodic perturbations and because of
secular perturbations slowly changing the MOID through the Lidov-Kozai
cycle \citep{gronchi:proper_crossing}. The empirical finding that the
impact probability and the probability of a close approach are
proportional, at least as a mean over thousands of cases, would
indicate that the MOID can be modeled as a random variable. However, a
rigorous mathematical formulation for this is not yet available.

\section{Comparison with JPL results}
\label{sec:JPL}

We performed a global comparison between the results of \clomon-2 and
Sentry. An asteroid-by-asteroid comparison is beyond the scope of this
paper, thus we present a statistical comparison using histograms like
those shown in Section~\ref{sec:results}. In particular, we take the
ensemble of all the virtual impactors computed by \clomon-2 and Sentry
at the same epoch (April 2018). Figure~\ref{fig:hist_CLOMON} refers to
the results of \clomon-2 and Figure~\ref{fig:hist_Sentry} to the ones of
Sentry: both plots represent the number of virtual impactors
$\mathcal{N}$ as a function of the inverse of the impact probability
$IP$.

\begin{figure}[t!]
    \centering
    \includegraphics[width=0.7\columnwidth]{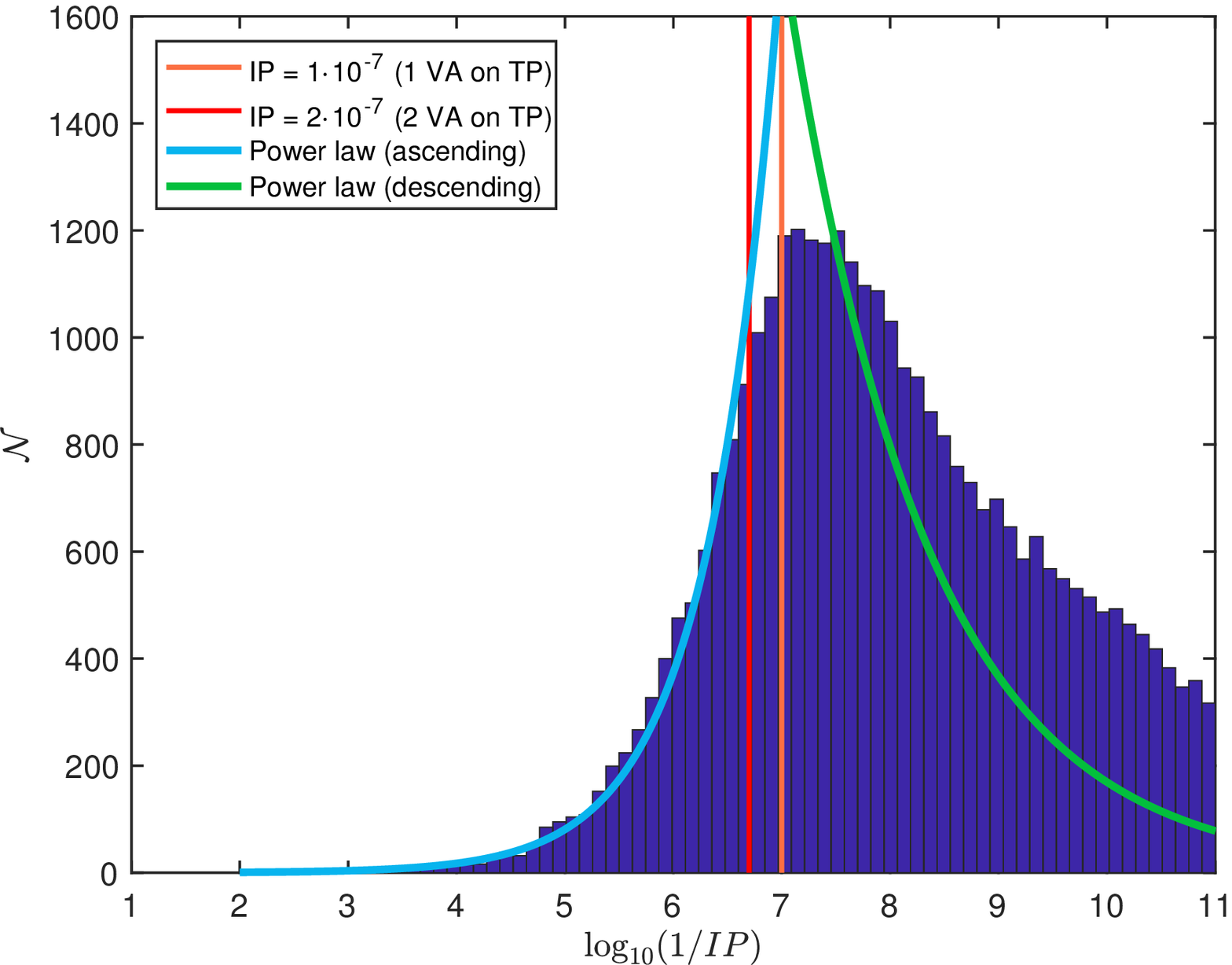}%
    \caption{Histogram of the number of virtual impactors in the
      NEODyS Risk List as a function of the inverse of the impact
      probability (as of April 2018).}
    \label{fig:hist_CLOMON}
    \vspace{0.5cm}
    \includegraphics[width=0.7\columnwidth]{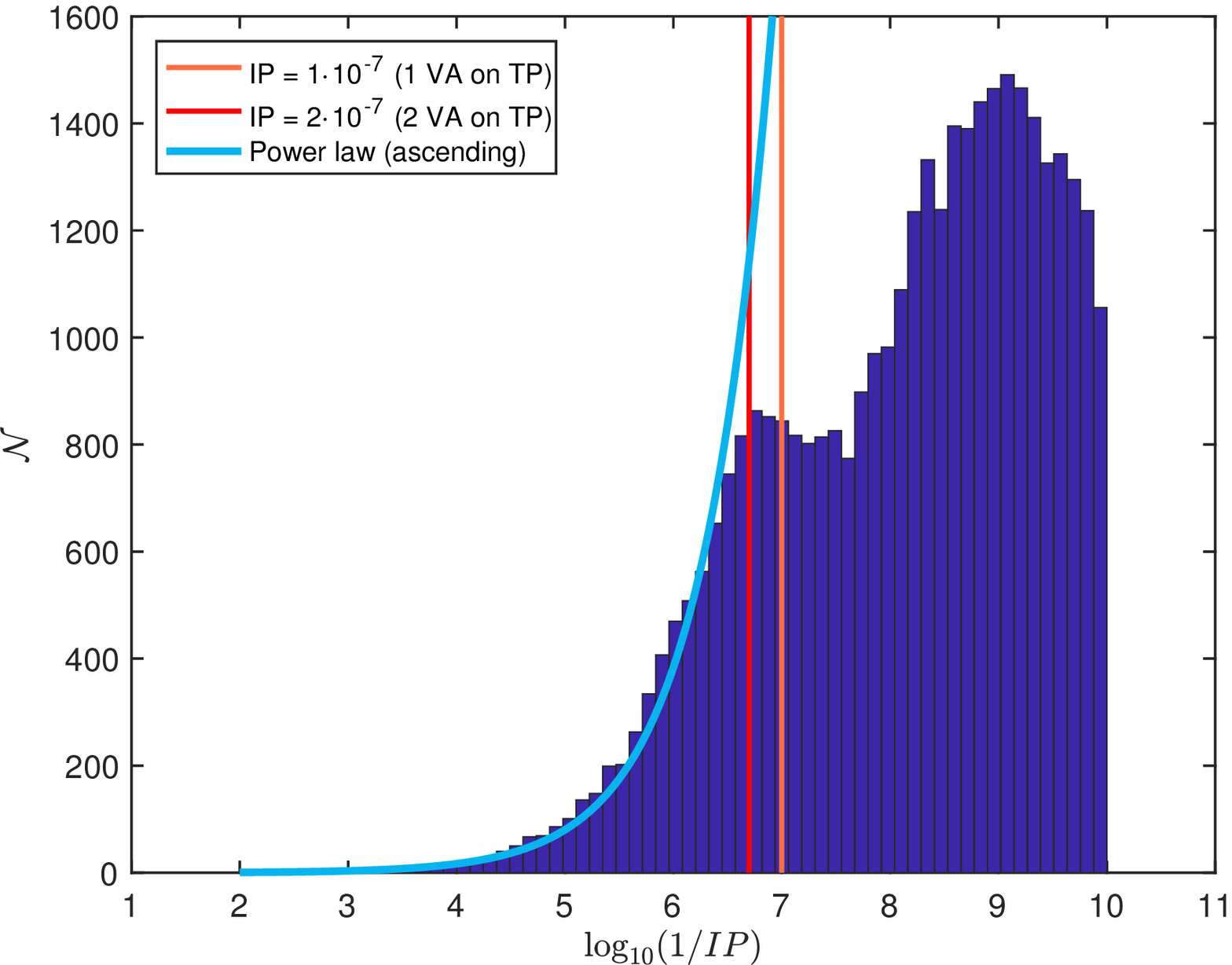}%
    \caption{Histogram of the number of virtual impactors in the
      Sentry Risk List as a function of the inverse of the impact
      probability (as of April 2018).}
    \label{fig:hist_Sentry}
\end{figure}

Both plots show a very good agreement in the ascending part up to
$IP\simeq 2\cdot IP^* = 2\cdot 10^{-7}$ (vertical red line). To
strengthen this argument we also performed a linear fit of the
histogram contour for $IP>IP^*$, as in
equation~\eqref{eq:ascending}. The exponent of the power-law resulted
to be $\alpha\simeq 0.664$ for the \clomon-2 results and $\alpha\simeq
0.679$ for the Sentry results. The number of VIs obtained by Sentry at
$IP\simeq 2\cdot IP^*$ is somewhat lower than our number and it is the
cause for the difference in the values of the multiplicative constant
$c_1$ of equation~\eqref{eq:ascending} obtained from both fits.

\noindent For much lower impact probabilities, the two plots show some
differences. The loss of efficiency in the region between the two
vertical lines is a common feature, but our histogram is increasing
whereas the one related to Sentry begins to slightly decrease. Overall
the behavior of the two systems in the detection of VIs with
$IP>IP^*$ is very well consistent, which was one of the goals of the
improvements in our system. The biggest difference is in the right
parts of the histograms: our plot is strictly decreasing,
corresponding to the fact that the number of virtual impactors grows
but the probability to detect them simultaneously decreases, whereas
the histogram related to Sentry shows a peak around $IP\simeq
10^{-9}$. This behavior might be explained by differences in the
computation techniques, in particular in the treatment of the off-LOV
virtual impactors \citep{milani:clomon2}, which usually have very low
impact probabilities. We plan to work on these differences in
collaboration with the JPL team.

In conclusion, the global comparison confirms a very good agreement
between the two systems. The differences we found are mostly explained
by technicalities of the methods used in the impact monitoring
computations and anyway mostly regard very low probability VIs. Of
course this does not exclude differences, in particular in the
computed $IP$ for each virtual impactor, which necessarily arise
because we are currently using two different error models for the
astrometric observations, also taking into account that these error
models are incomplete for lack of metadata. The fact that the results
for $IP>IP^*$ are very similar in terms of the number of VIs found is
a very significant result, precisely because two different error
models have been used.

\section{Conclusion and future work}
\label{sec:conc}

In this paper we reported on two improvements of our impact monitoring
system \clomon-2 with respect to what was described in
\citep{milani:clomon2}. The first one was a correct handling of the
cases in which a return on the target plane of the Earth includes two
instances of the same virtual asteroid. This was done by using a
recursive splitting of the showers, thus radically eliminating these
duplications. The second was to decrease the impact probability
corresponding to the generic completeness, which was previously $\simeq
4\cdot 10^{-7}$, to $1\cdot 10^{-7}$. We did not achieve this result
by brute force, that is by using four times more virtual asteroids,
but by using a sampling of the Line Of Variations optimized by a
uniform probability for each segment, at least for the portion closer
to the nominal solution.

\noindent Both improvements have been implemented in the operational
software and fully tested, by recomputing the entire risk list, that
is the asteroids known to have VIs. Note that both the improvements we
have implemented were removing differences between the algorithm used
in the Sentry system of JPL and our \clomon-2. Thus, having
implemented these two improvements, we were for the first time able to
perform a full statistical comparison between the global results of
\clomon-2 and those of Sentry, since the two systems should now be
giving more similar results.

When we first produced this kind of histograms of the number of VIs
found in all the risk list asteroids, as a function of variables such
as $1/IP$ and the stretching $S$, we found empirically that the number
of detected virtual impactors with $IP>IP^*$ appeared to grow
according to a power-law, proportional to $IP^{-2/3}$. As shown in the
figures of this paper, we have tested that this result, numerically
quite accurate, was obtained with the risk list as computed with
different sampling of the LOV and with different values of $IP^*$
(compare Figure~\ref{fig:cum1} and Figure~\ref{fig:cum2}), as computed
at different dates (compare Figure~\ref{fig:cum2} and
Figure~\ref{fig:hist_CLOMON}), and for risk lists at the same date but
computed with different software and different astrometric error
models (compare Figure~\ref{fig:hist_CLOMON} and
Figure~\ref{fig:hist_Sentry}). Thus we are lead to believe that we
have experimentally found a fractal property of the set of the initial
conditions leading to impacts in the chaotic dynamical system of
planet crossing asteroids. We must admit we do not yet have a model
explaining this power-law. We suspect it is related to the power-law
by which the cumulative number of VIs within a time $t_{rel}$ from the
first observed close approach grows proportional to the power-law
$t_{rel}^3$, for which we have found a number-theoretical
argument. However, for the connection between the two power-laws we
have not yet found a model, which we suspect to hide in properties of
the chaotic orbits of near-Earth asteroids.

\noindent Still, the use of the empirical law $IP^{-2/3}$ allows us to
explore, for the first time, the effective completeness of the impact
monitoring systems: of course this completeness cannot extend beyond
the generic completeness limit $IP^*$. The results are encouraging, in
that both \clomon-2 and Sentry are not just statistically consistent,
but also consistent with the empirical power-law, down to an $IP\simeq
2\cdot 10^{-7}$. This is a significant achievement, because we never
had a ``ground truth'' against which to assess our performance.  Since
the two impact monitoring systems are currently using two different
error models for astrometric observations, by using the standard
argument that the difference between the last two models can be
considered an estimate of the inaccuracies remaining in the last one,
this indicates also robustness of our results with respect to the
astrometric error model.

The discussion above clearly indicates the directions we should move
in our future work. First, we would like to close the gap between the
generic completeness at $IP^*=10^{-7}$ and the effective completeness
(resulting from the comparison with the power-law) at $IP\simeq 2\cdot
10^{-7}$. In principle we know how to do this by densification of the
cases in which a return contains too few points on the target plane,
but we would like to find a solution which is not brute force and this
requires some effort, but appears feasible. Second, we need to
investigate in depth the issue of the VI histogram power-law, to
understand if indeed it is a fractal property and provide at least an
approximate model explaining it, possibly starting from the success in
explaining the power-law with respect to time. This requires some new
idea, thus we are not able to claim that we shall solve this problem,
but we shall try. Also other researchers are welcome to try.

\section*{Acknowledgments}
A.~Del Vigna and A.~Chessa acknowledges support by the company
SpaceDyS. Part of this research was conducted under European Space
Agency contract No. 4000113555/15/DMRP ``P2-NEO-II Improved NEO Data
Processing Capabilities''.

We thank NASA-JPL for making available on their web risk page the data
on the virtual impactors found by Sentry, and also for the assistance
in downloading the full list. We thank Davide Farnocchia (JPL) for
useful discussions, in particular for the comparison between our
system and theirs.

We thank the referees (S. Chesley and an anonymous) for suggesting
improvements which we think have made easier the understanding of this
paper, including the new Figures~\ref{fig:shower_2000SG344_bef}
and~\ref{fig:split_ret}.

\appendix
\section{Decomposition of returns with duplications}
\label{app:dec_subret}

In this appendix we provide a detailed description of the procedure to
decompose a return with duplications into sub-returns, giving a
mathematical proof of completion.

A return $\calR$ is given by a contiguous LOV segment, that is the
indices of the corresponding virtual asteroids are consecutive. Let
$\calI_\calR$ be this set of indices. Let $n_{\mathcal R}$ the number
of distinct close approaches of the return $\calR$. We rigorously
define the return to be
\[
    \mathcal R \coloneq \{ (i_k,t_k) \,:\, i_k\in
    \calI_\calR\}_{k=1,\,\ldots,\,n_{\calR}},
\]
considering $\calR$ as the set of couples given by the LOV index and
the corresponding closest approach time. We also assume that the
sequence of times $(t_k)_{k=1,\,\ldots,\,n_{\calR}}$ is
non-decreasing, that is $t_{k+1}\geq t_k$ for all $k$. We now suppose
that the return contains a duplication, that is there exist
$k_1,k_2\in \{1,\,\ldots,\,n_{\calR}\}$ such that $k_1\neq k_2$ and
$i_{k_1}=i_{k_2}$, in such a way that the return has to be further
decomposed.

For $1\leq s\leq n_{\calR}$ define $I_s\coloneq \{i_k \,:\, 1\leq
k<s\}$. We now want to recursively define the sequence $(s_n)_{n\geq
  1}$ of the beginning points of the sub-showers. Let $s_1=1$ and, for
$n\geq 0$ define
\[
    s_{n+1} = \min_{s_n<s\leq n_{\calR}} \{s \,:\, i_s \in
    I_s\setminus I_{s_n}\}
\]
or $s_{n+1}=n_{\calR}$ in case the minimum does not exist because the
set is empty. By definition $(s_n)_{n\geq 1}$ is a non-decreasing
sequence. Since $s_n\leq N$ for all $n\geq 1$, there exists $n_s$ such
that
\[
    s_1 < s_2 < \cdots < s_{n_s}\quad\text{and}\quad s_n=n_\calR\
    \text{for } n>n_s.
\]
The integer $n_s$ is the number of sub-showers in $\calR$. For $1\leq
n\leq n_s$ let $\calI_n\coloneq I_{s_{n+1}}\setminus I_{s_n}$ be the
set of the indices of each sub-shower. Note that the indices in
$\calI_n$ are pairwise distinct by construction. Moreover, the
collection $\{\calI_n\}_{1\leq n\leq n_s}$ is a partition of
$\{1,\,\ldots,\,n_\calR\}$. The sub-showers of $\calR$ are thus
defined to be
\[
    \calS_n \coloneq \{(i_k,t_k) \,:\, k\in \calI_n\}
\]
for $1\leq n\leq n_s$ and
\[
    \calR = \bigsqcup_{n=1}^{n_s} \calS_n.
\]
There is no guarantee that the indices in a sub-showers are
consecutive, thus each sub-shower is further decomposed in sub-returns
by taking contiguous LOV segments. Rigorously, for each $1\leq n\leq
n_s$ there exist an integer $r_n\geq 0$ and a collection of $r_n$
subsets $\{\calR_{n,m}\}_{m=1,\,\ldots,\,r_n}\subset \calS_n$ such
that the indices of $\calR_{n,m}$ are a maximal set of consecutive
numbers among the indices of $\calS_n$ and such that
$\calS_n=\bigsqcup_{m=1}^{r_n} \calR_{n,m}$. In this way we obtain
\[
    \calR = \bigsqcup_{n=1}^{n_s} \bigsqcup_{m=1}^{r_n} \calR_{n,m},
\]
the decomposition of $\calR$ in sub-returns. Given this formal
procedure, since the initial number of close approaches contained in
starting return is finite, it is proven that the procedure has a
finite number of steps, at the end of which there are no duplicate
points.

\section{The Farey sequence}
\label{app:farey}
The \emph{Farey sequence} $\mathcal{F}_n$ of order $n$ is the
ascending sequence of irreducible fractions between $0$ and $1$ whose
denominators do not exceed $n$. That is, $\frac hk \in \mathcal{F}_n$
if $0\leq h\leq k\leq n$ and $(h,k)=1$\footnote{The symbol $(h,k)$,
  with $h$ and $k$ integers, denotes their greatest common divisor.},
with the numbers $0$ and $1$ included in the form $\frac 01$ and
$\frac 11$. For instance, the elements of $\mathcal{F}_5$ are the
following:
\[
    \frac 01,\ \frac 15,\ \frac 14,\ \frac 13,\ \frac 25,\ \frac 12,\
    \frac 35,\ \frac 23,\ \frac 34,\ \frac 45,\ \frac 11.
\]
By definition $\mathcal{F}_{n-1} \subseteq \mathcal{F}_{n}$ for all
$n>1$, and the smallest Farey sequence is that of order $n=1$
\[
  \mathcal{F}_1 = \left\{ \frac 01, \frac 10\right\}.
\]

The property of the Farey sequence which is important for the result
in Section~\ref{sec:nVI_time} concerns the asymptotic behavior of the
number of elements of $\mathcal{F}_n$. Its proof is can be found in
\cite{hardy_wright}. First we have to recall the definition of the
\emph{Euler totient function} $\varphi$. For each integer $n\geq 1$ we
define $\varphi(n)$ as the number of positive integers $\leq n$ and
relatively prime with $n$. For instance, $\varphi(1) = 1$, $\varphi(6) = 2$,
$\varphi(33) = 20$. There is a simple equation to evaluate $\varphi$:
\[
  \varphi(n)=n\prod_{p\mid n} \left(1-\frac{1}{p}\right),
\]
where the product is extended to all the prime numbers that divide
$n$.

\begin{thm}\label{thm:farey1}
  The length of the Farey sequence of order $n$ is
  \[
    |\mathcal{F}_n| = 1 + \sum_{k=1}^n \varphi(k).
  \]
\end{thm}

The following theorem gives the asymptotic behavior of
$|\mathcal{F}_n|$ as $n\rightarrow +\infty$. It can be proved by
Theorem~\ref{thm:farey1} and \cite[Theorem~330]{hardy_wright}.

\begin{thm}\label{thm:farey2}
  As $n\rightarrow +\infty$
  \[
    |\mathcal{F}_n| \sim \frac{3}{\pi^2}n^2.
  \]
\end{thm}

Using the Stoltz-Cesaro theorem we can prove the following statement
on the accumulation of $|\mathcal{F}_n|$.
\begin{thm}\label{thm:farey3}
    As $n\rightarrow +\infty$
    \[
      \sum_{k=1}^n |\mathcal{F}_k| \sim \frac{1}{\pi^2}n^3.
    \]
\end{thm}
\begin{proof}
    The sequence $b_n=n^3$ is increasing and $\lim_{n\rightarrow
      +\infty} b_n = +\infty$. The thesis follows from
    Theorem~\ref{thm:farey2} and from the Stoltz-Cesaro theorem
    since
    \[
    \dfrac{\sum_{k=1}^{n+1} |\mathcal{F}_k|-\sum_{k=1}^n
    |\mathcal{F}_k|}{(n+1)^3-n^3} =
    \frac{|\mathcal{F}_{n+1}|}{3n^2+3n+1} =
    \frac{|\mathcal{F}_{n+1}|}{(n+1)^2}\frac{(n+1)^2}{3n^2+3n+1}
    \rightarrow \frac{1}{\pi^2}.
    \]
\end{proof}

Analogous results hold in case we consider the irreducible fractions
between any two integers $r$ and $s>r$ and with denominator not
exceeding $n$. Let us denote that set of fraction with
$\mathcal{F}_n^{r,s}$.
\begin{lemma}
    Each unit interval $[\ell,\ell+1]$ contains exactly as many
    irreducible fractions with denominator not exceeding $n$ as
    $[0,1]$. That is, $|\mathcal{F}_n| =
    |\mathcal{F}_n^{\ell,\ell+1}|$.
\end{lemma}
\begin{proof}
    It suffices to prove that there exists a bijection between
    $\mathcal{F}_n$ and $\mathcal{F}_n^{\ell,\ell+1}$. If $\frac hk\in
    \mathcal{F}_n$, then consider $\frac hk + \ell$: it belongs to
    $[\ell,\ell+1]$ and it is irreducible since $(h,h+k\ell)=1$
    (otherwise $(h,k)>1$). Thus $\frac hk+\ell\in
    \mathcal{F}_n^{\ell,\ell+1}$. Analogously, if $\frac hk\in
    \mathcal{F}_n^{\ell,\ell+1}$ then $\frac hk-\ell$ belongs to
    $[0,1]$ and it is irreducible.
\end{proof}

\begin{thm}\label{thm:farey4}
    For $s>r$ integers, as $n\rightarrow +\infty$
    \[
    |\mathcal{F}_n^{r,s}| \sim \frac{3(s-r)}{\pi^2}n^2
    \quad\text{and}\quad
    \sum_{k=1}^n |\mathcal{F}_k^{r,s}| \sim \frac{s-r}{\pi^2}n^3.
    \]
\end{thm}
\begin{proof}
    it is possible to write $[r,s]$ as the union of adjacent unit
    intervals as $[r,s] = \bigcup_{i=0}^{s-r} [r+i-1, r+i]$. As a
    consequence and also by applying the above lemma we get
    \[
    |\mathcal{F}_n^{r,s}| = (s-r)|\mathcal{F}_n|-(s-r-1).
    \]
    The thesis follows from Theorem~\ref{thm:farey2}.
\end{proof}

\bibliographystyle{elsarticle-harv} \bibliography{compl_IM_biblio}

\end{document}